\documentclass[acmsmall,screen,nonacm]{acmart}
\setcopyright{none}
\usepackage{tabularx}
\usepackage{framed}
\usepackage{pifont}
\usepackage{placeins}
\microtypesetup{expansion=true,stretch=40,shrink=40}

\usepackage{amsmath,amssymb,amsthm}
\usepackage{mathtools}
\usepackage{booktabs}
\usepackage{tabularx}
\usepackage{graphicx}
\usepackage{listings}
\usepackage{xcolor}
\usepackage{tikz}
\usetikzlibrary{positioning,calc,arrows.meta,fit,backgrounds,chains}

\usepackage{stmaryrd}

\theoremstyle{plain}
\newtheorem{theorem}{Theorem}
\newtheorem{proposition}{Proposition}
\newtheorem{lemma}{Lemma}
\newtheorem{corollary}{Corollary}
\theoremstyle{definition}
\newtheorem{definition}{Definition}
\ifdefined\remark\else\newtheorem{remark}{Remark}\fi

\newcommand{\Sem}{\ensuremath{\mathcal{S}}}
\newcommand{\Terms}{\ensuremath{\mathrm{Terms}}}
\newcommand{\Sorts}{\ensuremath{\mathrm{Sorts}}}
\newcommand{\Ops}{\ensuremath{\mathrm{Ops}}}
\newcommand{\Gam}{\ensuremath{\Gamma}}
\newcommand{\Gamns}{\ensuremath{\Gamma_{\mathrm{ns}}}}
\newcommand{\dom}{\ensuremath{\mathrm{dom}}}
\newcommand{\Names}{\ensuremath{\mathrm{Names}}}

\newcommand{\proj}[1]{\ensuremath{\llbracket #1 \rrbracket}}
\newcommand{\Ginst}[1]{\ensuremath{#1[\Gam]}}
\newcommand{\Lang}[1]{\ensuremath{L(#1)}}
\newcommand{\esc}[1]{\ensuremath{\llbracket #1 \rrbracket}}
\newcommand{\slotname}[1]{\ensuremath{\texttt{\%#1\%}}}
\newcommand{\queryslot}{\ensuremath{\mathrm{query}_{\mathrm{slot}}}}
\newcommand{\renderalt}{\ensuremath{\mathrm{render\_alternation}}}

\newcommand{\pol}{\ensuremath{\pi}}
\newcommand{\refn}{\ensuremath{\sqsubseteq}}
\newcommand{\palign}{\ensuremath{\asymp}}
\newcommand{\scopesafe}{\ensuremath{\mathrm{ScopeSafe}(\Gam)}}
\newcommand{\refs}[1]{\ensuremath{\mathrm{refs}(#1)}}
\newcommand{\wallop}{\ensuremath{\tau_{\Gam}}}
\newcommand{\freebits}{\ensuremath{\mathrm{free\_bits}}}

\newcommand{\gproj}{\textsc{gproj}}
\newcommand{\inductor}{\textsc{TemplateInductor}}

\newcommand{\sort}{\ensuremath{s}}

\lstdefinestyle{gbnf}{basicstyle=\ttfamily\small,breaklines=true,
  keywordstyle=\bfseries,columns=fullflexible,showstringspaces=false,
  aboveskip=4pt,belowskip=2pt}
\newif\ifincludeappendix
\includeappendixtrue
\title{Decode-Time Grammars}
\subtitle{Constrained LLM Generation over a Refinement Order of Grammar Fragments}

\author{Shuoming Zhang}
\affiliation{\institution{SKLP, Institute of Computing Technology, Chinese Academy of Sciences}\city{Beijing}\country{China}}
\email{zhangshuoming21s@ict.ac.cn}

\author{Ruiyuan Xu}
\affiliation{\institution{SKLP, Institute of Computing Technology, Chinese Academy of Sciences}\city{Beijing}\country{China}}
\email{xuruiyuan23s@ict.ac.cn}

\author{Haofeng Li}
\affiliation{\institution{SKLP, Institute of Computing Technology, Chinese Academy of Sciences}\city{Beijing}\country{China}}
\email{lihaofeng@ict.ac.cn}

\author{Qiuchu Yu}
\affiliation{\institution{SKLP, Institute of Computing Technology, Chinese Academy of Sciences}\city{Beijing}\country{China}}
\email{yuqiuchu19@mails.ucas.ac.cn}

\author{Yangyu Zhang}
\affiliation{\institution{SKLP, Institute of Computing Technology, Chinese Academy of Sciences}\city{Beijing}\country{China}}
\email{zhangyangyu19b@ict.ac.cn}

\author{Chunwei Xia}
\affiliation{\institution{University of Leeds}\city{Leeds}\country{United Kingdom}}
\email{C.Xia@leeds.ac.uk}

\author{Xiaobing Feng}
\affiliation{\institution{SKLP, Institute of Computing Technology, Chinese Academy of Sciences}\city{Beijing}\country{China}}
\email{fxb@ict.ac.cn}

\author{Chenxi Wang}
\affiliation{\institution{SKLP, Institute of Computing Technology, Chinese Academy of Sciences}\city{Beijing}\country{China}}
\email{wangchenxi@ict.ac.cn}

\author{Huimin Cui}
\affiliation{\institution{SKLP, Institute of Computing Technology, Chinese Academy of Sciences}\city{Beijing}\country{China}}
\email{cuihm@ict.ac.cn}

\author{Jiacheng Zhao}
\authornote{Corresponding author.}
\affiliation{\institution{SKLP, Institute of Computing Technology, Chinese Academy of Sciences}\city{Beijing}\country{China}}
\email{zhaojiacheng@ict.ac.cn}

\authorsaddresses{Authors' contact information: Shuoming Zhang,
zhangshuoming21s@ict.ac.cn; Jiacheng Zhao (corresponding author),
zhaojiacheng@ict.ac.cn. SKLP, Institute of Computing Technology, Chinese
Academy of Sciences, Beijing, China.}

\begin{document}
\begin{abstract}
Large language models now write a growing share of the world's code,
increasingly inside agents and serving systems that compile, execute, or
dispatch generated code without line-by-line review. This works well for
mainstream languages but remains brittle for low-resource programming surfaces
such as domain-specific languages, custom library APIs, and command-line
tools. Even under grammar-constrained decoding, a model can still produce
references invalid in the current environment: a buffer never declared, a
column absent from the schema, a function the library does not provide, or an
unsupported CLI option.

This paper introduces \emph{decode-time grammars}: grammar fragments
instantiated during generation from a runtime environment $\Gam$. A
region-specific policy $\pol(\sort,\Gam)$ selects a fragment for each hole,
and $\tau_\Gam$ replaces open reference positions with $\Gam$-typed slots
whose candidates are exactly the names, fields, APIs, or options available at
that point. Newly generated declarations enter $\Gam$ before later regions are
decoded, so the constraining grammar can depend on the prefix already
generated. This ensures not only grammatical correctness but also semantic
correctness, by preventing references to undefined symbols.

We formalize grammar fragments as environment-indexed grammars ordered by
refinement, prove No-Ghost soundness for $\Gam$-slotted fragments, show that
refinement preserves this support-set guarantee, and characterize the boundary
of mask-enforceable properties. We implement the approach in \gproj{} with
offline grammar induction and online policy resolution. Across TileLang, SQL,
and P4, with models from 0.6B to 236B parameters, \gproj{} eliminates ghost
references by construction at moderate overhead over standard constrained
decoding.

\end{abstract}

\maketitle
\renewcommand{\shortauthors}{Zhang et al.}
\section{Introduction}
\label{sec:intro}

Large language models have become a practical way to write code. They complete
programs in mainstream languages fluently \cite{codex2021arxiv}, and they
increasingly do so inside agents and serving stacks where generated code may
reach a compiler, a database, or a shell before a human reviews each line
\cite{swebench2024iclr,vllm2023sosp}. In this regime, validity is not only a
matter of model quality. It is also a matter of what the generation procedure
can rule out.

This paper studies a recurring failure in low-resource programming surfaces:
references to objects that do not exist in the current environment. These
surfaces include tensor-kernel DSLs such as TileLang and Triton
\cite{tilelang2025arxiv,triton2019mapl}, compiler IRs \cite{mlir2021cgo},
data-plane languages such as P4 \cite{p4lang2014ccr}, custom library APIs, and
command-line tools. A representative TileLang failure is
\texttt{T.gemm(ictensor\_q\_newQ, ...)}: the first operand names a buffer that
was never declared. The line is syntactically plausible, but
\texttt{@tilelang.jit} fails when resolving the name. We call such an
environment-invalid reference a \textbf{ghost reference} (\autoref{fig:motivation}, left): an undefined name to
the PL reader, and a mechanically decidable subclass of hallucination to the LLM
literature \cite{ji2023halluc}.\footnote{This is unrelated to ghost state in
program verification. Verification ghosts exist for proofs but not at runtime;
our ghosts appear in the generated program but name nothing.}

Ghost references are not simply typos. They often arise from negative transfer:
the model writes the nearby dialect, API version, or tool interface it expected,
rather than the one present in the target environment. For example, when asked
for a TileLang kernel, DeepSeek-V4-Flash \cite{deepseekv4arxiv} repeatedly
reaches for TVMScript's \texttt{T.Buffer}, a neighboring spelling deprecated in
TileLang for \texttt{T.Tensor} (\autoref{sec:eval-strongmodel}). This kind of error
can persist even as models become more capable, because the wrong continuation
is a fluent and high-probability one.

Grammar-constrained decoding is the natural place to look for a remedy. Existing
engines mask the next-token distribution against a context-free grammar, so only
grammatically legal continuations remain in the support set
\cite{xgrammar2025mlsys,syncode2025tmlr,outlines2023arxiv,gcd2023emnlp}. This
guarantees syntactic validity with respect to the given grammar. However, a fixed
grammar usually treats reference positions as open classes: a kernel operand is
an \texttt{identifier}, a SQL column is a \texttt{name}, and a command-line flag
is an \texttt{option}. Such productions admit both valid and invalid references.
The missing condition is environment-dependent: the name must be in scope, the
column must be in the schema, the function must exist in the library version, or
the option must be supported by the current tool. Each of these facts is available at generation time; none is expressible by a
grammar fixed independently of the environment.

\begin{figure}[t]
  \centering
  \includegraphics[width=\columnwidth]{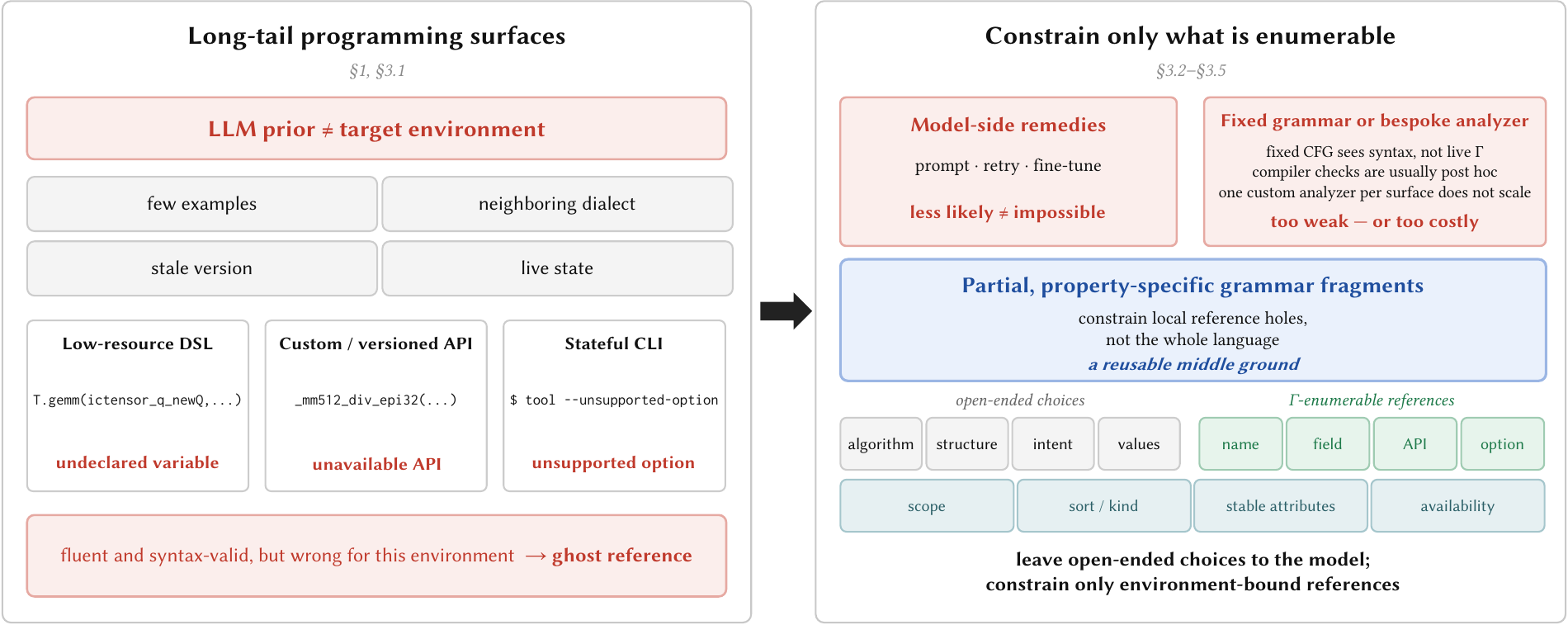}
  \caption{Motivation and design insight. Long-tail programming surfaces can
  produce syntax-valid but environment-invalid ghost references. Partial,
  property-specific grammar fragments constrain only environment-bound
  references, while leaving open-ended program choices to the model.}
  \label{fig:motivation}
\end{figure}

Our key move is to make the evolving generated prefix an input to the
constraint itself. A \emph{decode-time grammar} is a grammar fragment
instantiated during generation from a runtime environment $\Gam$. The current
prefix determines $\Gam$, while the expected sort of the next hole determines
which environment facts are relevant. The key operator $\wallop$ then replaces
an open reference position with a $\Gam$-typed slot whose candidates are exactly
the names, fields, APIs, or options available at that point. The resulting
fragment specializes the recognizer before the hole is decoded. Unlike
prompting, fine-tuning, or reweighting, which at best make the ghost unlikely
while leaving it sampleable, this specialization removes it from the support set.

Unlike conventional approaches that generate source code under a fixed language definition, our approach incorporates environments into the code-generation process itself: generated
declarations extend $\Gam$, an LLM sketch exposes the next typed hole, a
per-region policy selects a grammar fragment, and that fragment is
JIT-instantiated from the current $\Gam$ into the token-level recognizer. The
result therefore acts directly on the support of the continuation,
whose emitted declarations may in turn update the analysis state for later
regions.

We instantiate grammar fragments at decode time from the evolving environment $\Gam$. Declaration regions extend $\Gam$with newly generated names, and later reference regions are constrained by the updated support set. This runtime construction admits exactly the references available in the current environment, ruling out ghost references while avoiding the blocking that arises when declarations are fixed in advance. Thus, decode-time grammar instantiation provides a support-set guarantee that previous fixed-support precompiled grammar families cannot achieve without sacrificing non-blocking generation.

We separate environment-bound correctness from open-ended program decisions (\autoref{fig:motivation}, right). Decode-time grammars enforce the former by ensuring that references are drawn from the current $\Gam$, thereby eliminating mechanical reference errors that are enumerable from the environment. The model remains responsible for choices that require program understanding, such as the algorithm, tiling strategy, query intent, and command sequence. This division of labor is capability-orthogonal: improving the model can improve these open-ended choices without weakening the support-set guarantee enforced by the mask. Thus, decode-time grammars provide a precise, stable slice of correctness rather than a claim of whole-program correctness.

One programming construct may be exposed through grammar fragments of different
strengths. For the same \texttt{T.copy} operand position, one fragment may admit
an arbitrary identifier, another may admit only names in $\Gam$, a third may
filter those names by sort or shape, and a fourth may pin one particular
operand. These fragments preserve the surrounding semantic construct while
exposing different amounts of structure to the mask. We model such fragments as
semantic-to-syntax projections into a shared semantic domain $\Sem$
\cite{krogmeier2023oopsla,hazel2017popl}, ordered by refinement $\refn$. A
per-region policy $\pol(\sort,\Gam)$ selects the fragment to use at each hole.
Tight fragments provide stronger guarantees; coarser fragments preserve freedom
or provide fallback when a stronger property is unavailable. This order, and
the $\wallop$ edge where the no-ghost guarantee switches on, are summarized in
\autoref{fig:paper-map}.

We implement the idea in \gproj{}, an online masked executor that decodes across
nested region holes while maintaining $\Gam$, and in \inductor{}, an offline
inductor that constructs grammar fragments from small corpora under a validation
gate. Our evaluation spans low-resource DSLs, custom library APIs, and
command-line tools, with models ranging from 0.6B parameters to a 236B frontier
model. On TileLang, SQL, and P4, $\Gam$-typed slots turn reference safety from a
model-dependent outcome into a construction-level guarantee, while preserving
the model's freedom to choose the surrounding program. On custom library usage,
the induced gates reject invented functions, intrinsics, and build options. For
command-line tool use, the tool's own self-description and live state provide
$\Gam$, making unsupported actions absent from the support set. Across model
sizes, the guarantee transfers with the mask: stronger models improve the
semantic sketch, while the environment-bound reference layer remains fixed.
Across evaluations on TileLang, SQL, and P4 with models
ranging from 0.6B to 236B parameters, every $\Gam$-typed arm is ghost-free
by construction; on SQL, \gproj{} improves execution match from
$76\%$ to $100\%$. This guarantee incurs moderate
serving cost: relative to XGrammar, \gproj{} reduces throughput by
10.6--17.8\%, while the end-to-end reduction relative
to unconstrained decoding averages $17.3\%$
.

\paragraph{Contributions.}
\begin{itemize}\setlength{\itemsep}{2pt}
\item We introduce \textbf{decode-time grammars}: grammar fragments
specialized from a runtime environment $\Gam$ while an LLM-generated program is
still being constructed. The operator $\wallop$ turns open reference positions
into $\Gam$-typed slots, compiling prefix-dependent analysis facts directly
into the decoder's support set and removing ghost references by construction.

\item We formalize grammar fragments as \textbf{semantic-to-syntax projections}
ordered by refinement. We prove No-Ghost soundness for $\Gam$-slotted
fragments, show that refinement preserves this support-set guarantee, and
characterize the boundary of mask-enforceable properties (\autoref{sec:formalization}).

\item We prove a necessity result: no finite family of precompiled grammars
with fixed reference supports, under any prefix-reading dispatch, can be both
ghost-free and non-blocking. Thus exact $\Gam$-dependent reference support must
be synthesized during decoding; grammar instantiation is our direct realization
of this necessary runtime object.

\item We build \gproj{}, an online masked executor, and \inductor{}, an offline
inductor for constructing validated grammar-fragment policies. Across
low-resource DSLs, custom library APIs, command-line tools, and models from
0.6B to 236B parameters, the system demonstrates a transferable division of
labor: the LLM supplies the semantic sketch, while the decode-time grammar
guarantees environment-bound references. This yields oracle-checked
improvements in generation quality, with reference safety transferring across
model sizes and programming surfaces.
\end{itemize}

\paragraph{What is new.}
Grammar-masked decoding conventionally supplies a recognizer to a masking
backend, while prior systems already enrich generation with prefix-dependent
completion engines, runtime-populated parsers, and programmable semantic
properties \cite{synchromesh2022iclr,top2025colm,chopchop2026popl}.
Our contribution is an
\emph{environment-indexed grammar policy} in which the generated prefix itself
drives construction of the recognizer used for its continuation. The fragment
library is induced offline from corpus and schema evidence; online, an LLM
sketch identifies the next typed analysis boundary, live-$\Gam$ facts specialize
the selected fragment, and the resulting recognizer is compiled into the token
mask before that region is emitted. Analysis therefore acts not only as a
checker or completion oracle, but as part of the source-construction procedure.

The refinement order makes the space of such recognizers explicit: the policy
can select stronger or weaker fragments per region, track where guarantees
become active, and fall back when a stronger analysis is unavailable.
\autoref{sec:related-work} positions the closest works. Synchromesh establishes prefix-derived masks;
Tree-of-Parsers demonstrates runtime-populated semantic
slots; Krogmeier and Madhusudan make grammar spaces objects of synthesis
\cite{krogmeier2023oopsla}; and
ChopChop supplies programmable semantic realizability. We add a refinement-ordered and inducible
grammar policy, its online specialization from an evolving program environment,
and an explicit support-set guarantee with a necessity result: no
fixed-support family provides the same ghost-free, non-blocking behavior
(Prop.~3). Building on XGrammar-style JIT masking
\cite{xgrammar2025mlsys}, we introduce this runtime-$\Gam$ analysis layer.
Symbol-table scope is the first concrete instance; the broader abstraction is
\emph{JIT construction of an environment-indexed recognizer for generated code}.

\section{Background}
\label{sec:background}

This section fixes the vocabulary used throughout the paper: token masks,
runtime environments, holes, and the induction primitives used to construct
grammar fragments. Related work is discussed in \autoref{sec:related-work}.

\subsection{Token masking and the support set}
\label{sec:bg-masking}

An autoregressive model decodes one token at a time from a distribution over
vocabulary $V$ \cite{vaswani2017attention}. Before sampling,
\emph{constrained decoding} uses a formal constraint, such as a regex or CFG,
to assign zero probability to tokens that cannot legally continue the current
prefix. The admitted tokens form the \textbf{support set}, represented by a
boolean \textbf{mask}. Unlike prompting, fine-tuning, or reweighting, which only
change token probabilities, masking removes continuations from the support.
Thus every guarantee in this paper concerns what the decoder can emit, not what
the model prefers.

In practice, grammar masks are implemented by engines such as XGrammar, SynCode,
and Outlines \cite{xgrammar2025mlsys,syncode2025tmlr,outlines2023arxiv}. The
engine maintains a parser state for the prefix and intersects the grammar's
legal next terminals with the tokenizer vocabulary. If the engine is exact, any
completed string emitted under grammar $G$ lies in $\Lang{G}$ by construction.
This is a guarantee about the decoding procedure, not about termination or
semantic correctness of the generated program.

We use \textbf{\freebits} as a post-hoc diagnostic of mask freedom:
$\freebits=\sum_t \log_2 |M_t|$, where $M_t$ is the set of tokens admitted by
the mask at step $t$ in a completed decode. It is an explanatory property of the
realized mask trace, not an optimization objective or a property of model
probabilities: fixed strings have $\freebits=0$, while unconstrained regions
approach $\log_2|V|$ per token.

\subsection{From fixed grammars to environment-bound constraints}
\label{sec:bg-lineage}

Masking against a fixed grammar is now a standard way to enforce structured
generation for off-the-shelf models \cite{gcd2023emnlp}. It also appears behind
search-side or declarative front-ends such as PICARD and LMQL
\cite{picard2021emnlp,lmql2023pldi}. Our executor builds on this substrate.

Several lines of work extend fixed-syntax masking by adding semantic checks,
reshaping the distribution, or repairing invalid generations after the fact
\cite{synchromesh2022iclr,typeconstrained2025pldi,chopchop2026popl,
top2025colm,grammaraligneddecoding2024neurips,smc-llm-control2025iclr,
itergen2025iclr}. Our focus is different: when a reference position has a finite
candidate set determined by the current environment, we remove invalid
candidates from the support set before sampling. \autoref{sec:related-work} compares these approaches
in detail.

\subsection{Runtime environments across DSLs, libraries, and tools}
\label{sec:bg-edsl}

Our primary setting is generation for low-resource, environment-bound
programming surfaces. In embedded DSLs such as TileLang and Triton
\cite{tilelang2025arxiv,triton2019mapl}, a generated region is interpreted
inside a host program, so valid references depend on buffers and values declared
earlier. In SQL over Spider schemas \cite{spider2018emnlp}, valid column and
table names are determined by the database schema. In P4 programs
\cite{p4lang2014ccr}, valid fields depend on declarations in the current packet
processing context.

We write the runtime environment as $\Gam$. Depending on the surface, $\Gam$
contains in-scope names, sorts, shapes, schema entries, API members, command-line
flags, or tool state. The common feature is that well-formedness depends on
membership in $\Gam$: a name must have been declared, a column must exist, an API
must be available, or an option must be supported. This condition is not captured
by a fixed CFG production such as \texttt{identifier}; Remark~1 formalizes this
point for declaration consistency.

The same abstraction applies beyond DSLs. For library-API usage, $\Gam$ is the
target library's API surface: functions, intrinsics, enum constants, options, and
version-specific names. For command-line tool use, $\Gam$ is supplied by the
tool's self-description together with mutable state such as branches, build
targets, scripts, and available subcommands \cite{react2023iclr,toolformer2023neurips}.
The three settings differ in where $\Gam$ comes from, but the invalid-reference
failure mode and the decode-time remedy are the same.

\subsection{Holed terms, templates, and induction primitives}
\label{sec:bg-holes}

A partially written program is a term with holes. We use the standard vocabulary
of typed holes and sketch-based synthesis
\cite{hazel2017popl,sketch2006asplos}: a \textbf{template} is a holed term, and
each \textbf{hole} has an expected \textbf{sort} $\sort$. In our setting, holes
are filled by masked decoding rather than by human editing or solver search.

The offline component constructs grammar fragments from small corpora using two
standard primitives. First, \emph{anti-unification} computes least general
generalizations for aligning examples
\cite{plotkin1970generalization,babble2023popl}. Second, grammar induction is
retargeted from language approximation
\cite{arvada2021ase,hygenar2025acl,grammarprompting2023neurips} to fragments
that are executable by the masking engine and can be instantiated from $\Gam$.
The resulting fragments are validated against the online executor before use
(\autoref{sec:construction}).

\begin{figure*}[t]
  \centering
  \includegraphics[width=\textwidth]{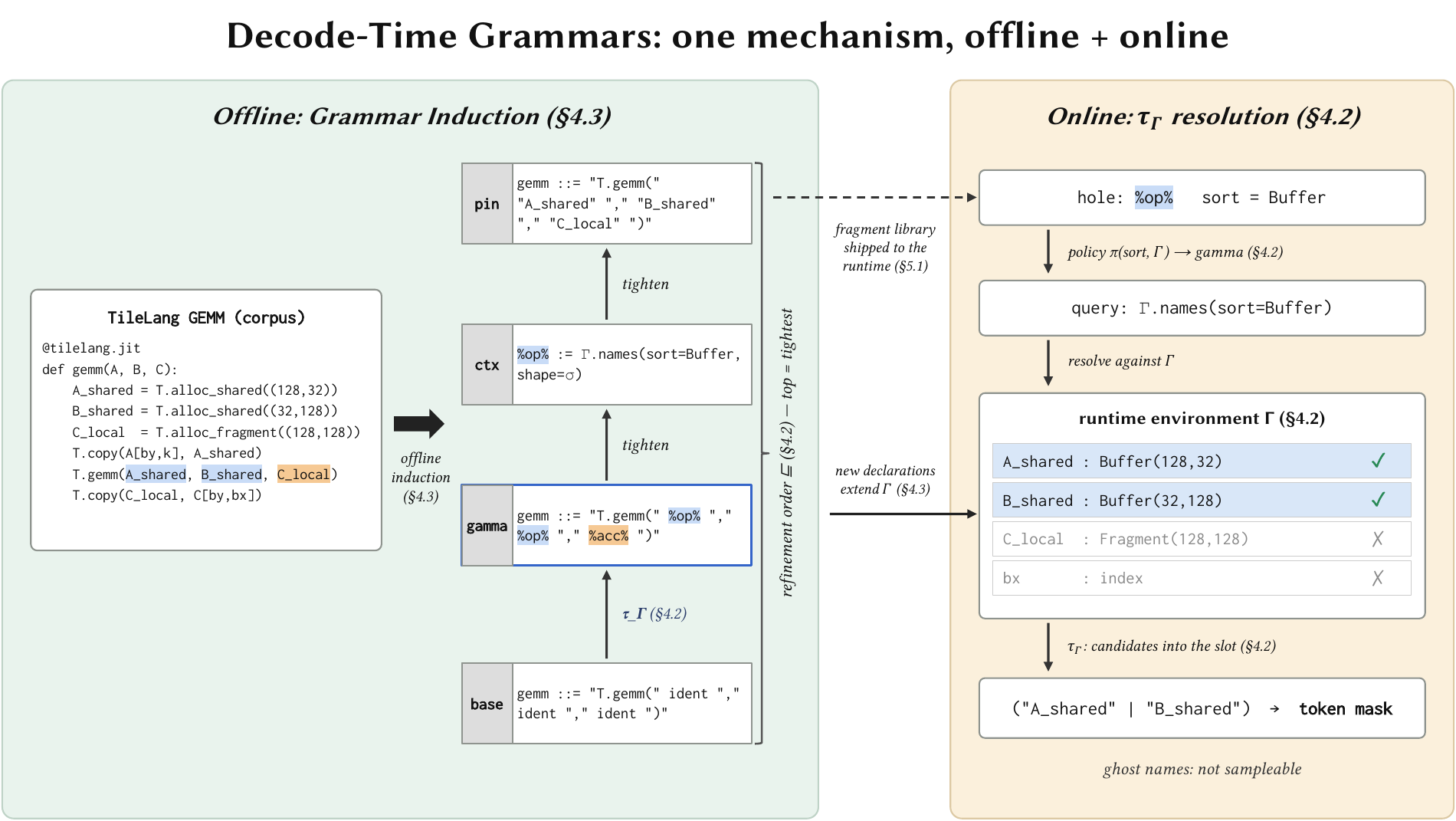}
  \caption{Decode-time grammars, end to end. Left: the offline inductor (\autoref{sec:construction})
  lowers a TileLang corpus into a ladder of grammar fragments ordered by
  refinement (\autoref{sec:sdcg}); colored slots mark the operand holes they came from, and
  \wallop{} is the edge where the no-ghost guarantee switches on. Right:
  at decode time the policy \pol{} picks a rung, and \wallop{} tightens its open
reference positions into slots instantiated from the runtime environment $\Gam$ --- only \texttt{Buffer}-sort names survive
  into the token mask; ghost names are not sampleable (Theorem~1).}
  \label{fig:paper-map}
\end{figure*}

\section{Motivation}
\label{sec:motivation}

This section motivates decode-time grammars through one recurring failure class:
references that are syntactically valid but invalid in the current environment.
We first show the failure on TileLang and SQL
(\autoref{sec:motivation:failure}), then explain why fixed grammars and
model-probability remedies do not remove it
(\autoref{sec:motivation:fixed-grammar}). The fix exposes the two design objects
used later: a refinement order of grammar fragments
(\autoref{sec:motivation:many-grammars}) and a division of labor between mask
guarantees and model choice
(\autoref{sec:motivation:freedom}--\autoref{sec:motivation:boundary}).
\autoref{fig:trace-cliff} shows the key move on one real decode.

\begin{figure*}[t]
  \centering
  \includegraphics[width=0.88\textwidth]{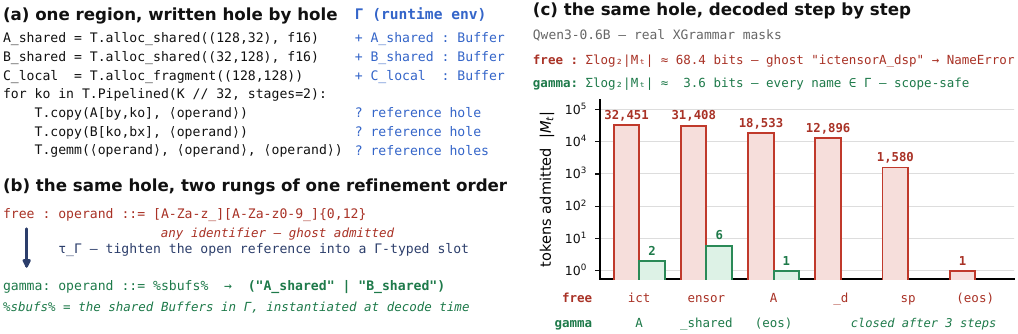}
  \caption{The \wallop{} move on one real decode (Qwen3-0.6B, real XGrammar
  masks; declaration holes are pinned for presentation, reference holes are
  decoded freely). \textbf{(a)} A GEMM region is generated hole by hole; each
  declaration extends $\Gam$ before later holes are instantiated. \textbf{(b)}
  The same reference hole at two rungs: an open-identifier fragment
  (\texttt{free}) and a $\Gam$-typed slot (\texttt{gamma}). \textbf{(c)} Under
  \texttt{free}, the mask admits 32{,}451 first-step tokens and the model emits
  the ghost \texttt{ictensorA\_dsp}; under \texttt{gamma}, the support contains
  only declared buffers, so every emitted name is scope-safe by construction.}
  \label{fig:trace-cliff}
\end{figure*}

\subsection{The failure: negative transfer, not ignorance}
\label{sec:motivation:failure}

Consider Qwen3-0.6B, a deliberately small and locally deployable model, asked to
generate the inner body of a TileLang GEMM kernel \cite{tilelang2025arxiv}. When
matrix-multiply operands are decoded as free identifiers, a representative output
is:

\begin{lstlisting}[language=Python]
# inside an @tilelang.jit kernel; A_shared, B_shared, C_local were declared above
T.gemm(ictensor_q_newQ, B_shared, C_local)
\end{lstlisting}

The line is syntactically valid: it parses, and a generic
\texttt{buffer ::= identifier} grammar accepts it. The error is environmental:
\texttt{ictensor\_q\_newQ} was never declared, so \texttt{@tilelang.jit} raises
\texttt{NameError}. In this tier, every generated program fails to compile, and
each failure has the same form. The same pattern appears in SQL. For Spider
queries over a fixed schema \cite{spider2018emnlp}, the model can emit
\texttt{SELECT nonexistent\_col FROM singer ...}: well-formed SQL that fails
because the column is absent. On our difficulty-controlled subset, with the
column position decoded as a free identifier, 90/90 generated queries contain a
ghost column (\autoref{sec:evaluation}).

These failures are not caused only by lack of data. A stronger model may know a
nearby dialect better and transfer it to the wrong target
\cite{negtransfer2005nips}. In our TileLang
setting, an unconstrained frontier model reaches for TVMScript's
\texttt{T.Buffer} on every sample, although TileLang uses \texttt{T.Tensor}
(\autoref{sec:eval-strongmodel}). The program looks fluent, but for the wrong
environment. This is the narrow failure class we target: the model often has the
right high-level intent, but names the wrong concrete object.

The class extends beyond variables. A ghost reference may be a missing intrinsic
(\texttt{\_mm512\_div\_epi32}; \autoref{sec:eval-invariance}), an undeclared P4 field, a wrong keyword
argument, a version-specific API, or an unsupported command-line option. The
family is not anecdotal: naming-class hallucination is a top category in
execution-verified taxonomies \cite{codehalu2025aaai}, and hallucinated
\emph{package} names recur predictably enough to be weaponized as a
supply-chain attack \cite{packagehalluc2025usenix}. The
common test is whether the legal candidates can be enumerated from $\Gam$: prior
declarations, a schema, an API surface, a version, or tool state. Positions that
do not pass this test---argument values, call protocols, numerical correctness,
termination---fall outside our guarantee and remain the model's responsibility
(\autoref{sec:motivation:boundary}; formal boundary in \autoref{sec:construction}).

\subsection{Why neither a fixed grammar nor a prompt clears it}
\label{sec:motivation:fixed-grammar}

Tightening an ordinary CFG does not solve the problem. The rule
\texttt{buffer ::= identifier} already captures the surface shape of a buffer
reference. What it cannot express is that the identifier must name a buffer
declared earlier in this kernel. That condition depends on the generated prefix
and is not context-free (Remark~1, \autoref{sec:formalization}).

Model-probability remedies have the complementary limitation. Prompting,
fine-tuning, in-context instructions, self-repair \cite{selfrepair2024iclr},
and reweighting can make a
legal continuation more likely, but they do not remove an illegal one from the
support. An instruction such as ``use \texttt{T.Tensor}, not \texttt{T.Buffer}''
can reduce the error rate, and in-context examples can bias the model toward the
target dialect \cite{fewshot2020neurips}; neither provides a guarantee. As the environment description
grows---more buffers, more API rules, more version constraints---the relevant
constraint competes for attention with the rest of the prompt, and the model may
still sample a fluent but invalid continuation. Against a strong wrong prior, such as a nearby dialect spelling, a weight-axis
remedy must overcome the model's preference. A mask-side remedy does not:
removing \texttt{T.Buffer} or an undeclared buffer from the support costs the
same regardless of the model's probability for it.

The decode-time fix is to instantiate the reference production from $\Gam$ when
the reference is generated. If $\Gam$ currently contains
\texttt{A\_shared}, \texttt{B\_shared}, and \texttt{C\_local}, the operand
production becomes a finite choice among those names. This step is \wallop{}:
an open reference position becomes a $\Gam$-typed slot. In
\autoref{fig:trace-cliff}, the open identifier admits a ghost, while the
$\Gam$-typed slot admits only declared buffers. In evaluation, toggling this one
edge is the single-variable ablation behind the no-ghost cliff: scope-safety
moves from 0\% to 100\% at the $\Gam$-typed rung (\autoref{sec:evaluation}). The claim is narrow:
$\Gam$-slots rule out ghost references, while semantic correctness is checked by
external oracles.

\subsection{One DSL, many reasonable grammars}
\label{sec:motivation:many-grammars}

The previous subsection did not identify one correct grammar. It identified a
tighter fragment, useful because $\Gam$ was available. This is the structure we
use throughout: one DSL can expose the same semantic construct through fragments
of different strengths. For example, a TileLang copy can be written as:

\begin{lstlisting}[language={}]
T.copy(src, dst)                       # API-as-call:   generic function-call grammar; loose
copy src -> dst                        # API-as-syntax: operation lifted to dedicated syntax
CopyStmt(src: Buffer, dst: Buffer)     # pinned:        typed semantic slots exposed
\end{lstlisting}

These fragments can project to the same operation in $\Sem$, while exposing
different amounts of structure to the mask. Which fragment is useful depends on
the region being generated. A high-level skeleton should preserve model freedom;
a reference slot should be tied to $\Gam$; a field or option may be pinned to a
small enumeration. Thus the object we need is not a single grammar, but a policy
over fragments.

We order these fragments by refinement. Moving downward accepts fewer surfaces
and exposes more structure; moving upward provides fallback when a stronger
fragment is unavailable. In this order, \wallop{} is the edge where the no-ghost
guarantee turns on. Once a reference position is below that edge, further
tightening cannot re-admit an undeclared name; \autoref{sec:formalization} proves this as a refinement
preservation property.

\subsection{Two axes of freedom}
\label{sec:motivation:freedom}

A fixed recognizer makes constrained decoding look one-dimensional: looser or
tighter. Program generation needs two axes. One is semantic or structural
freedom: which algorithm to write, which loops to use, what tiling strategy to
choose, or what query intent to express. The other is surface or reference
freedom: which buffer, column, API token, field, or command option is available
in the current environment.

The desired point is not the tightest possible grammar. A fully pinned template
can remove ghosts, but it also removes the model's freedom to choose the
surrounding program. Decode-time grammars instead leave semantic structure open
where possible and tighten only environment-bound reference positions. The
No-Ghost theorem in \autoref{sec:formalization} is stated at this level: it does not require a fixed
program template, only that reference positions selected by the policy are
realized as $\Gam$-typed slots.

Using the post-hoc diagnostic from \autoref{sec:bg-masking}, the GEMM ladder shows
the intended pattern: tightening lowers \freebits{}, while scope-safety jumps at
the \wallop{} edge. Reference freedom is reduced where $\Gam$ can enumerate the
legal choices; structural freedom remains with the model. \autoref{sec:eval-wall}
reports the full tier-by-tier measurement; \autoref{sec:eval-sql} repeats the
analysis on SQL.

\subsection{The boundary as a division of labor}
\label{sec:motivation:boundary}

The boundary is a design choice, not a disclaimer. Decode-time grammars enforce
the part of correctness expressible as membership in $\Gam$: names in scope,
schema columns, declared fields, available APIs, and supported options. They do
not enforce semantic intent, functional correctness, numerical accuracy,
call-order protocols, or termination. Those properties require program
understanding or execution and remain with the model, with external oracles used
only for evaluation.

This split matches the empirical profile. Negative transfer is especially
damaging at reference positions: a nearby but invalid spelling can be fluent and
high probability. Model capability helps with open-ended semantic and structural
choices, as the steady improvement on execution-verified benchmarks---
SWE-bench's repository fixes, KernelBench's GPU kernels, and Terminal-Bench's
command-line sessions---illustrates
\cite{swebench2024iclr,kernelbench2025arxiv,terminalbench2025}. These
improvements, however, carry no support-set guarantee, and fluent models still
produce both semantic errors and environment-invalid references. We target the
latter and, more generally, the mechanically enumerable constraints expressible
by $\Gam$. Our own data exhibits the same separation at both ends of the scale:
the unconstrained 0.6B model already reaches 76\% execution match on our SQL
subset (\autoref{sec:eval-sql}), while the frontier model sketches plausible
TileLang programs but fails on the target API surface
(\autoref{sec:eval-strongmodel}). The mask handles the mechanical, enumerable
error class; the model handles open-ended program choice.

The same split changes both sides of the pipeline. The producer builds an order
of grammar fragments rather than a single recognizer, and the consumer decodes
under a per-region policy $\pol(\sort,\Gam)$ rather than under one fixed grammar.
The next section makes these objects precise and proves the support-set
guarantees they provide.

% \begin{lstlisting}[language=Python]
% # inside an @tilelang.jit kernel; A_shared, B_shared, C_local were declared above
% T.gemm(ictensor_q_newQ, B_shared, C_local)
% \end{lstlisting}

% \begin{lstlisting}[language={}]
% T.copy(src, dst)                       # API-as-call:   generic function-call grammar; loose
% copy src -> dst                        # API-as-syntax: operation lifted to dedicated syntax
% CopyStmt(src: Buffer, dst: Buffer)     # pinned:        typed semantic slots exposed
% \end{lstlisting}

\section{Formalization and Construction}
\label{sec:formalization}

This section develops decode-time grammars in full, following the two panels
of \autoref{fig:paper-map}: \autoref{sec:formal-prelim} fixes the objects; \autoref{sec:sdcg} treats one hole,
a single fragment decoded under runtime semantic information, with its
guarantees; \autoref{sec:construction} composes holes into whole programs, induces the fragments
offline, and shows why exact reference support must be synthesized online; \autoref{sec:clients} shows the
induced grammars on three client surfaces.

In one sentence: the executor builds the symbol table and resolves against it
\emph{inside} the decoding loop --- declarations extend $\Gam$ as they are
generated, references resolve against the current $\Gam$ through the mask ---
so scope analysis runs synchronously with generation, and scope safety holds
by construction (Theorem~1, Corollary~2) rather than by later checking.

The soundness results establish reference safety in the decoder's support set,
over the name-to-sort layer of $\Gam$. Compilation and functional behavior are
orthogonal properties measured by the external oracles in \autoref{sec:evaluation}.

\subsection{Preliminaries}
\label{sec:formal-prelim}

\begin{table}[!t]
\caption{Notation used in the formalization.}
\label{tab:notation}
\centering
\scriptsize
\setlength{\tabcolsep}{4pt}
\resizebox{\textwidth}{!}{
\begin{tabular}{@{}llll@{}}
\toprule
$\Sigma$, $\Sigma^*$ & tokenizer vocabulary; token strings & $\Gam$ & runtime environment $\Names \rightharpoonup (\mathit{sort},\mathit{attrs})$ \\
$\sort$; $\Sorts$, $\Ops$ & a sort; sort and operator sets of $\Sem$ & $\Gamns$; $\dom(\Gam)$ & active name$\to$sort snapshot; declared-name set \\
$\Sem=(\Sorts,\Ops)$ & semantic domain (many-sorted term algebra) & \texttt{\%slot\%}; $\queryslot$ & a $\Gam$-slot; its query $\Gam\to\mathcal{P}_{\mathrm{fin}}(\Names)$ \\
$\mathtt{Buffer}\langle\rho,\mu,\sigma\rangle$ & buffer sort: role $\rho$, memory space $\mu$, shape $\sigma$ & $\Ginst{G}$ & $G$ with each slot instantiated against $\Gam$ \\
$\Terms_s(\Sem)$ & well-typed $\Sem$-terms of sort $s$ & $\renderalt(C)$; $\esc{c}$ & escaped, parenthesized alternation of $C$; escaped literal \\
$(\sort,\Gam)$ & a hole: expected sort, current environment & $\Lang{G}$ & surface strings $G$ accepts \\
$G$; $\mathrm{anchor}(G)$ & a grammar fragment; its anchor sort & $\refs{d}$; $\refs{w}$ & names at reference positions (one derivation; all) \\
$\proj{\cdot}_G$ & projection $\Lang{G}\to\Terms_{\mathrm{anchor}(G)}(\Sem)$ & $\scopesafe$; $L_{\mathrm{scope}}$ & scope-safe strings at $\Gam$; declaration-consistent language \\
$\refn$; $\palign$, $\approx_{\Sem}$ & refinement preorder; projection alignment (mod $\approx_{\Sem}$) & $\wallop$ & tightening: an open reference position becomes a $\Gam$-slot \\
$\pol(\sort,\Gam)$ & per-region policy: selects the fragment for a hole & $\freebits$ & $\sum_t \log_2|M_t|$: freedom of one masked decode \\
\bottomrule
\end{tabular}}
\end{table}

The definitions below fix the objects over which the machine-checkable
premises of \autoref{sec:sdcg}--\autoref{sec:construction} are stated; \autoref{tab:notation} collects the
notation. Let $\Sigma$ be the tokenizer vocabulary and $\Sigma^*$ the set of
token strings.

\begin{definition}[Hole and policy]
A \emph{hole} is a pair $(\sort,\Gam)$ awaiting a grammar fragment. The
per-region policy $\pol$ maps $(\sort,\Gam)$ to a fragment. A \emph{leaf}
fragment has no sub-holes; a \emph{composite} fragment is a template whose
sort-anchored sub-holes are resolved by $\pol$, at decode time, to further
fragments (the grammar-call tree of \autoref{sec:construction}).
\end{definition}

Granularity is not a separate notion: a statement-, block-, or
whole-kernel-level fragment differs only in its anchor sort, which is what lets
one program mix constraint tiers per region. In \autoref{fig:paper-map},
each operand position of \texttt{T.gemm} is a hole of sort \texttt{Buffer}.

\paragraph{Semantic domain.}
The semantic domain $\Sem=(\Sorts,\Ops)$ is a many-sorted term algebra
\cite{goguen1977initial}. Sorts include, for example,
$\mathtt{Buffer}\langle\rho,\mu,\sigma\rangle$ for role, memory space, and shape,
as well as $\mathtt{Dim}$, $\mathtt{DType}$, $\mathtt{Stmt}$, and
$\mathtt{Kernel}$. Operators include typed constructors such as
$\mathtt{Gemm}:\mathtt{Buffer}^3\to\mathtt{Stmt}$. We write
$\Terms_s(\Sem)$ for well-typed terms of sort $s$.

$\Sem$ is deeper than surface syntax: whether an operand names a previously
declared, sort-compatible buffer is neither an AST invariant nor context-free
(Remark~1), so scope-safety is a property of generated references at the level
of $\Sem$.

\paragraph{Remark 1 (membership in $\Gam$ is not context-free).}
``Every reference names a prior declaration'' is not CFG-expressible over any
name space with a two-letter subspace: intersecting with a regular declare/use
skeleton reduces it to the copy language $\{w\#w\}$, which is not context-free.
This is the same phenomenon behind Floyd's result for ALGOL~60
\cite{floyd1962algol}. We write $L_{\mathrm{scope}}$ for the
declaration-consistent language; no single fixed CFG recognizes it exactly. Full
argument: Appendix~A.8. This is the first half of necessity; Proposition~3
(\autoref{sec:construction}) closes the other half.

\begin{definition}[Fragment and projection]
A \emph{fragment} is a grammar piece $G$ equipped with an anchor sort
$\mathrm{anchor}(G)$ and a syntax-directed projection
$\proj{\cdot}_G : \Lang{G} \to \Terms_{\mathrm{anchor}(G)}(\Sem)$; for
ambiguous grammars, $\proj{\cdot}_G$ is read along a canonical derivation.
\end{definition}

A given semantic construct may have several fragments projecting to it, such as
\texttt{T.copy(src,dst)}, \texttt{copy src -> dst}, and
\texttt{CopyStmt(src:Buffer,dst:Buffer)}. This multiplicity is what the
refinement order organizes. The formal results use the projection through the
alignment condition in $\refn$; the evaluated ladders witness alignment by
same-carrier projection equality (Appendix~A.7).

\paragraph{Runtime environment.}
The runtime environment presented to a hole is a snapshot $\Gam : \Names \rightharpoonup (\mathit{sort},\mathit{attrs})$. Its name-to-sort projection is $\Gamns$, and $\dom(\Gam):=\dom(\Gamns)$. Bindings are write-once within each scope frame: declarations may extend a live frame, but do not change the sort of an existing binding in that frame. Lexical shadowing is represented later by overlaying distinct frames; consequently, the active snapshot may change when frames are pushed or popped, while each live frame and the persistent global layer grow monotonically. The single-hole soundness results require only the active snapshot at that hole.

\paragraph{$\Gam$-slots and instantiation.}
A fragment may contain slots written \texttt{\%slot\%}. Each slot carries a
query $\queryslot : \Gam \to \mathcal{P}_{\mathrm{fin}}(\Names)$.
A reference slot is well formed if $\queryslot(\Gam)\subseteq\dom(\Gam)$; the
common idiom is a filtered lookup such as
$\Gam.\mathrm{names}(\mathit{sort}=s,\dots)$.

Instantiation replaces each slot by an escaped alternation of its candidates:
$\Ginst{G} := G.\mathit{ebnf}\big[\slotname{slot} \mapsto
\renderalt(\queryslot(\Gam))\big]_{\forall\,\mathrm{slot}}$.
The renderer $\renderalt$ emits either one escaped literal $\esc{c}$ or a
parenthesized alternation $(\esc{c_1}\mid\cdots\mid\esc{c_k})$; empty candidate
sets raise an error. Parentheses and escaping are the only hygiene requirements
used by Lemma~1. The full concrete syntax and escaping definition are in
Appendix~A.1--Appendix~A.2.

\paragraph{References and scope safety.}
Each fragment marks its reference positions. A marked position is realized either
as an open identifier or as a $\Gam$-slot; $\wallop$ changes the former into the
latter. Literal or enumeration slots are not reference positions. Marking adequacy is a premise of Theorem~1: every name-yielding reference position must be marked. The inductor validates this premise against mined ghost substitutions (\autoref{sec:construction}).

For a derivation $d$, let $\refs{d}$ be the set of names yielded at marked
reference positions; for a string $w$, define
$\refs{w} := \bigcup_{d:\,\mathrm{yield}(d)=w}\refs{d}$.
The union over all derivations avoids any unambiguity assumption. The scope-safe
set at environment $\Gam$ is
$\scopesafe := \{\,w\in\Sigma^* \mid \refs{w}\subseteq\dom(\Gam)\,\}$.
For a fixed environment $\Gam$, $\scopesafe$ is the set of strings whose
references are contained in $\dom(\Gam)$. We also use $L_{\mathrm{scope}}$ for
the \emph{declaration-consistent language}: the set of complete strings in which
every reference names an earlier declaration in the same string. Remark~1 and
Proposition~3 concern this prefix-dependent language, not $\scopesafe$ for any
fixed $\Gam$.

\begin{definition}[Refinement]
For fragments $G_1,G_2$ with the same anchor sort,
\[
\begin{aligned}
G_1 \refn G_2
\quad:\Longleftrightarrow\quad&
\mathrm{anchor}(G_1)=\mathrm{anchor}(G_2)
\;\wedge\;
\proj{\cdot}_{G_1}\palign \proj{\cdot}_{G_2}
\\
&\wedge\;
\forall\Gam.\;
\Ginst{G_1}\ \mathrm{defined}
\Rightarrow
\big(
  \Ginst{G_2}\ \mathrm{defined}
  \wedge
  \Lang{\Ginst{G_1}}\subseteq\Lang{\Ginst{G_2}}
\big).
\end{aligned}
\]
\end{definition}

Thus $G_1$ refines $G_2$ when it accepts fewer surfaces at every environment,
while projecting to an aligned semantic construct. Alignment $\palign$ means
that on common surfaces the projections agree up to a DSL-specific equivalence
$\approx_{\Sem}$ on $\Terms(\Sem)$.

\begin{definition}[Tightening]
The operator $\wallop$ replaces an open reference position in a fragment by a
$\Gam$-slot with candidate set $\queryslot(\Gam)$.
\end{definition}

This single edge is what the paper toggles everywhere:
\autoref{fig:paper-map}'s ladder crosses it between \texttt{base} and
\texttt{gamma}; \autoref{sec:evaluation} measures it surface by surface.

\subsection{Semantic-Directed Constrained Generation}
\label{sec:sdcg}

Consider one hole of \autoref{fig:paper-map}: sort \texttt{Buffer}, at an
environment $\Gam$ declaring \texttt{A\_shared}, \texttt{B\_shared}, and
\texttt{C\_local}. The policy $\pol(\sort,\Gam)$ selects a fragment; the
fragment's slot query runs against the current $\Gam$, resolving
\texttt{\%op\%} to $\{$\texttt{A\_shared}, \texttt{B\_shared}$\}$;
instantiation splices these names into the grammar as an escaped alternation;
and the engine compiles $\Ginst{G}$ to a token automaton and decodes under the
mask, where an out-of-grammar token is absent rather than down-weighted
(\autoref{sec:bg-masking}). Three things could break this promise: the splice could leak,
creating strings that were never candidates (Lemma~1 rules this out); the
engine could diverge from the grammar semantics the proof uses (Assumption~1
pins that interface, which we validate by differential testing); and the guarantee
could turn out vacuous or stifling (Theorem~1 states what it gives,
Corollary~1 what it leaves free).

\paragraph{Assumption 1 (engine--denotation correspondence).}
For the fragment slice we instantiate, the GBNF denotation used in the proof agrees with the masking engine's parser, and the engine's token mask is an exact incremental recognizer of that denotation. We empirically validate this assumption over the complete syntactic construct vocabulary used by rendered slots---escaped literals and parenthesized alternation---together with adversarial representatives of their value space and token-level mask replay. The theorem remains explicitly modulo this engine-correspondence assumption. Details and test protocol are in Appendix~A.3.

\begin{lemma}[Escaping Lemma]
Let $C=\{c_1,\dots,c_k\}\subseteq\Names$ be finite and non-empty, and let
$A=\renderalt(C)$. Under the GBNF denotation, $L(A)=C$. Moreover, inlining $A$
into a concatenative production context
$r=\alpha\,\slotname{slot}\,\beta$ yields
\[
  L(r[A]) = L(\alpha)\cdot C\cdot L(\beta),
\]
the slot contributing no string outside $C$.
\end{lemma}

\begin{proof}[Proof sketch]
Escaping makes each $\esc{c_i}$ one literal terminal, and parentheses keep the
alternation local to the slot. The result is ordinary substitution of a finite
terminal set into a concatenative context
\cite{hopcroftullman1979,aho2006compilers}. Full proof: Appendix~A.2.
\end{proof}

Parenthesization prevents the usual \texttt{|}-split bug: splicing an
unparenthesized \texttt{"A" | "B"} into a larger production can create a
top-level alternative and admit strings that were never candidates.

\begin{theorem}[No-Ghost Soundness, single-hole; modulo Assumption~1]
Let a hole $(s,\Gam)$ have $\pol$ select a leaf fragment $G$, and suppose:
(i) every reference position in $G$ is a $\Gam$-slot;
(ii) every such slot has non-empty candidates;
(iii) every reference slot satisfies $\queryslot(\Gam)\subseteq\dom(\Gam)$.
Then
\[
  \Lang{\Ginst{G}}
  \subseteq
  \scopesafe
  =
  \{\,w : \refs{w}\subseteq\dom(\Gam)\,\}.
\]
Thus every string emitted by the mask under $\Ginst{G}$ references only declared
names.
\end{theorem}

\begin{proof}
Take any $w\in\Lang{\Ginst{G}}$, any derivation $d$ with yield $w$, and any
$n\in\refs{d}$. The name $n$ is yielded by a marked reference-position subtree.
By hypothesis, that position is a $\Gam$-slot instantiated as
$\renderalt(\queryslot(\Gam))$. By Lemma~1, the yielded name belongs to
$\queryslot(\Gam)$. By well-formedness,
$\queryslot(\Gam)\subseteq\dom(\Gam)$, so $n\in\dom(\Gam)$. Since $n$ and $d$
were arbitrary, $\refs{w}\subseteq\dom(\Gam)$.
\end{proof}

The theorem establishes support-set reference safety; compilation and
functional correctness remain properties for the external oracle. When a query
returns no candidates, the fallback policy below selects a coarser rung.

\begin{corollary}[Guarantee $\times$ flexibility]
At a reference hole satisfying Theorem~1, ghost names have zero support, and the
remaining name-level freedom is exactly the choice among
$\queryslot(\Gam)$. The model may still choose the wrong legal candidate; the
guarantee is only that every candidate is declared.
\end{corollary}

\begin{proof}[Proof sketch]
Immediate from Lemma~1 and Assumption~1: the slot language is exactly
$\queryslot(\Gam)$, and the mask realizes that language.
\end{proof}

\paragraph{One DSL, many rungs.}
The same hole can be filled at several strengths: an open identifier, the
declared set, a sort- or shape-filtered subset, or a pinned name
(\autoref{fig:paper-map}, left). Where $\Gam$ can say little, a loose
fragment preserves the model's freedom; where $\Gam$ can enumerate the legal
candidates, a tight fragment buys the guarantee; a program mixes both, region
by region, through $\pol$. The refinement order of \autoref{sec:formal-prelim} organizes this
choice, and two results make it safe to exercise at run time.

\begin{proposition}[Refinement preorder]
On fragments with the same anchor sort, $\refn$ is a preorder. Reflexivity is
immediate. Transitivity follows by chaining anchor equality, projection
alignment, and language inclusion at each environment where the fragments are
defined. Antisymmetry holds only modulo projection equivalence; arbitrary
meet/join existence is left open.
\end{proposition}

Proposition~1 makes the order well defined; Lemma~2 is what the runtime leans
on --- switching rungs never re-opens the proof.

\begin{lemma}[Refinement preserves support-set soundness]
Let $G_{\mathrm{tight}}\refn G_{\mathrm{loose}}$. If
$G_{\mathrm{loose}}$ is support-set sound at a hole $(s,\Gam)$, then
$G_{\mathrm{tight}}$ is also support-set sound wherever it is defined:
$\Lang{\Ginst{G_{\mathrm{tight}}}} \subseteq \Lang{\Ginst{G_{\mathrm{loose}}}}
\subseteq \scopesafe$.
Consequently, runtime switching among $\Gam$-slotted rungs
(\texttt{gamma}/\texttt{ctx}/\texttt{pin}) cannot re-admit a ghost. Falling above
\texttt{gamma} to a base grammar deliberately gives up the no-ghost guarantee.
\end{lemma}

\begin{proof}
The first inclusion is the language-inclusion conjunct of $\refn$; the second is
the assumed soundness of $G_{\mathrm{loose}}$. The consequence follows because
each $\Gam$-slotted rung is independently sound by Theorem~1, and tighter rungs
remain inside the language of looser sound rungs.
\end{proof}

A single fixed grammar, or one completion engine, gives no such
\emph{inter-grammar} guarantee: it secures only the one grammar it was built
for.

\paragraph{The tightening edge.}
By Definition~4, $\wallop$ turns an open reference position into a $\Gam$-slot.
When the candidates also lex as the original identifier form,
$\wallop(G)\refn G$; even when escaping takes the names outside that syntactic
comparison, Theorem~1 still gives support-set soundness for
$\Ginst{\wallop(G)}$. This is the edge where the no-ghost guarantee turns on.
On the induction ladder all rungs share a carrier production, so the language
inclusions needed for $\refn$ are witnessed directly by the slot renderer; we
need no general grammar-inclusion solver.

\paragraph{Fallback.}
Lemma~2 licenses a \emph{policy} over the order rather than a fixed choice.
When a rung is unavailable at a hole --- its slot query returns no candidates,
or confidence is low --- $\pol$ climbs one rung up, trading constraint
strength for progress while every emitted string stays grammatical. With
$\Gam$ still empty, a reference hole decodes freely under the base grammar;
once declarations land in $\Gam$, later reference holes are constrained to
them. One three-hole template shows the whole path (\texttt{\{:s\}} is a hole
of sort \texttt{s}):

\begin{lstlisting}[language={}]
SELECT {:ColRef} FROM t0;           -- Gamma empty: pi falls back to base; a free identifier
CREATE TABLE t1 ({c:ColDecl} INT);  -- the declaration extends Gamma
SELECT {:ColRef} FROM t1;           -- Gamma = {c}: pi re-tightens; support is exactly c
\end{lstlisting}

\noindent This is the \emph{graded guarantee} in mechanism form:
construction-level where $\Gam$ expresses the condition, graceful degradation
elsewhere. Degradation does not propagate: each hole
is instantiated from the $\Gam$ at its own decode, so $\pol$ re-tightens as
soon as declarations land (the third hole above). And the empty-query event
doubles as a diagnosis: raised before any token is sampled, it names the hole,
its sort, and the $\Gam$ snapshot --- the earliest mechanically decidable
evidence that the sketch omitted a declaration, which the two-stage pipeline
of \autoref{sec:method-two-stage} can route back to the sketch stage.

\subsection{Fragment Construction}
\label{sec:construction}

A whole program is decoded as many holes, one fragment at a time. Fragments
compose by grammar call; Theorem~1 lifts along the call tree to Corollary~2;
the fragments and the policy are induced offline under a hard gate; and
Proposition~3 shows why exact reference supports cannot be precompiled.

\paragraph{Grammar call and nesting.}
A skeleton grammar $G_{\mathit{top}}$ partitions the program into sort-anchored
region-holes, executed by \textbf{grammar call} with procedure-call semantics
(\autoref{fig:decode-time-interleave}). On entering a hole the executor
forms $c=(\sort,\Gam,\rho)$: the expected sort, the current environment, and
the region path $\rho$ on the call stack. The path lets $\pol$ select
differently for the same sort in different regions. The selected fragment is
\textbf{pushed} onto the grammar-call stack and \textbf{popped} on completion.
The masking engine itself has no such stack and no runtime-$\Gam$
instantiation (\autoref{sec:rel-infra}); both live outside the engine, in the executor. This is
where the grammar becomes \emph{self-extending}: the program, as it is being
written, extends the grammar that constrains its continuation --- the
grammar-level realization of the dynamic support Proposition~3 makes
necessary. Nesting goes to arbitrary depth, since a fragment may itself contain further region holes. For every completed generation, composition expands by finite substitution along an acyclic grammar-call tree. Ordinary CFG recursion inside an individual fragment retains its standard least-fixed-point semantics; recursive cycles through grammar calls are disallowed and additionally guarded by an engineering recursion cap. Nesting does not currently compose with the $\mathit{Times}$ repetition
combinator (\autoref{sec:method-limitations}).

\begin{figure*}[t]
  \centering
  \includegraphics[width=0.95\textwidth]{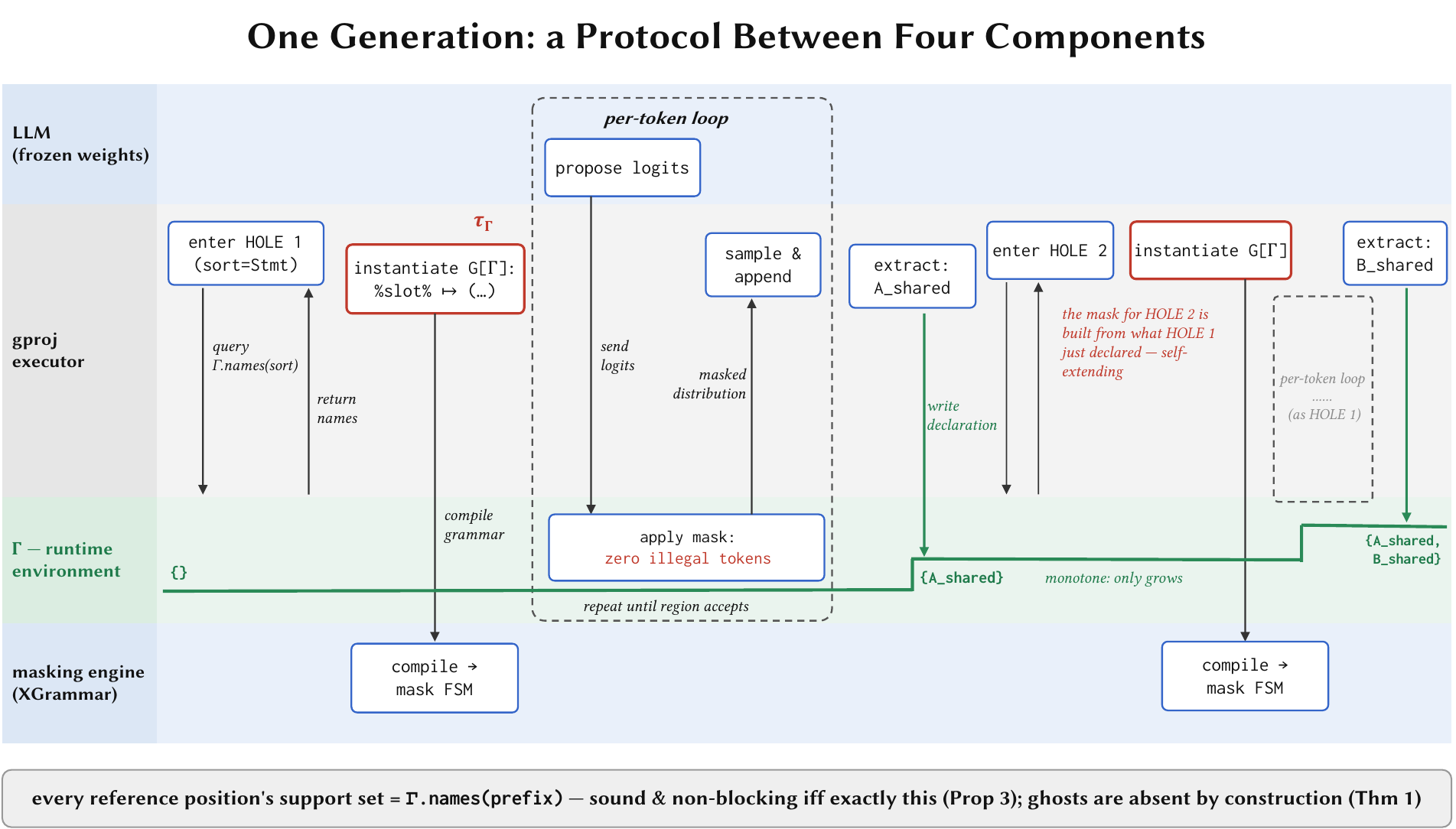}
  \caption{One generation, as a protocol between the four components. At each hole
  the executor queries $\Gam$, instantiates $\Ginst{G}$ (\wallop), and the engine
  compiles it to a token mask; the per-token loop zeroes illegal tokens; extracted
  declarations update the active $\Gam$ snapshot, monotonically within the
  depicted frame --- so \textsc{hole}~2's mask is
  built from what \textsc{hole}~1 just declared (self-extension).}
  \label{fig:decode-time-interleave}
\end{figure*}

\paragraph{Whole-program event model.}
Generation proceeds hole by hole. When a declaration hole completes, the
executor extracts the declared names and extends $\Gam$ \emph{before} any
later hole is instantiated --- the staircase of
\autoref{fig:decode-time-interleave}. In \autoref{fig:paper-map}, this
is what makes \texttt{A\_shared} a candidate for the \texttt{\%op\%} hole
decoded after it.
Declarations extracted from a completed hole are added to the current frame before any later hole in that frame or its descendants is instantiated; cross-region agreement uses the resulting active lexical overlay.
A $\Gam$ fixed before decoding (SQL) and one growing during it (TileLang) share this
one code path. Declaration-before-use is therefore guaranteed at hole
granularity, not necessarily within a single hole. Completed generations
unfold a finite grammar-call tree; non-completing decodes emit no completed
surface. A finite token budget can still strand a locally legal prefix. This
early-exit incompleteness is shared by all masking-based constrained decoding,
whose syntax guarantee is conditional on the decode reaching acceptance; it is
not fundamental, since budget-aware masking excludes, per step, the tokens
from which acceptance is no longer reachable within the remaining budget
\cite{dang2026mitigating}, and composes with our per-hole fragments. All
whole-program claims below are stated for completed generations.

\paragraph{Scope frames.}
The environment has a persistent global layer $\Gam^{g}$ and a stack of
region-local frames $\Delta_0,\ldots,\Delta_k$. The active environment at a
hole is their lexical overlay
$\Gam_k=\operatorname{overlay}(\Gam^{g},\Delta_0,\ldots,\Delta_k)$, with the
nearest frame resolving shadowed names. The persistent global layer and each live local frame grow monotonically, while the active overlay changes on frame push and pop; a grammar call pushes a local frame on entry and pops it on
return. Names declared in a frame are therefore visible to its descendants and
absent from every later sibling. Fragment boundaries coincide with block
boundaries, so the grammar-call tree is also the scope tree. The global and
lexical layers share the same slot-query interface over the active overlay.

\paragraph{Environment fidelity.}
For every generated prefix and active frame, each binding recorded in the active environment corresponds to a declaration visible at that point in the prefix, with its recorded sort and write-once attributes matching that declaration. Declaration extraction and frame management are part of the trusted interface and are required to preserve this invariant.

\begin{corollary}[Whole-program No-Ghost]
Let a completed generation $w$ under $\pol$ unfold a finite grammar-call tree.
For each reference occurrence $r$, let $h(r)$ be the leaf hole whose
instantiated fragment emits $r$, and let $\Gam_r$ be the active environment
when $h(r)$ is entered. If environment fidelity holds and every leaf fragment
selected by $\pol$ satisfies Theorem~1 at its active environment, then
\[
  \forall r\in\operatorname{RefOcc}(w).\quad
  \operatorname{name}(r)\in\dom(\Gam_r).
\]
Thus every emitted reference is in scope at its generation point. References
resolved to the persistent global layer also remain in
$\dom(\Gam^{g}_{\mathrm{final}})$.
\end{corollary}

\begin{proof}[Proof sketch]
Each reference occurrence belongs to a leaf fragment instantiated against
$\Gam_r$. Theorem~1 places the reference in that environment; later frame
updates or pops cannot affect this earlier resolution. Monotonicity of
$\Gam^{g}$ gives the final clause for global bindings. Environment fidelity
turns membership in the active environment into source-level lexical
visibility. Full proof: Appendix~A.10.
\end{proof}

At a name-complete reference position, the unfiltered declared-name query recovers the exact support exposed by ordinary symbol-table scope resolution. Corollary~2 gives the decode-time soundness direction, while Proposition~3(iii) characterizes exactness at such positions.

\begin{figure*}[t]
  \centering
  \includegraphics[width=0.96\textwidth]{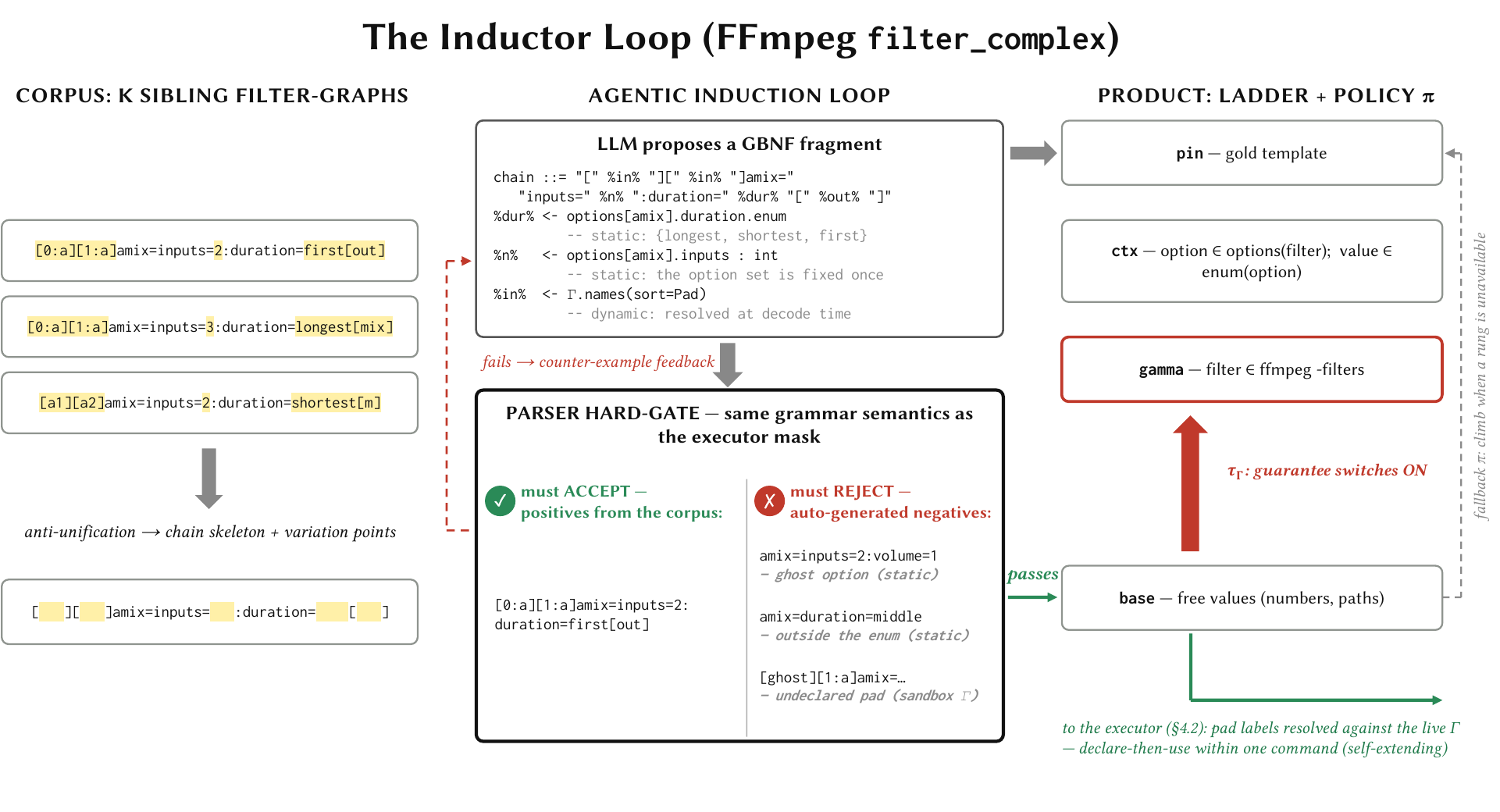}
  \caption{The \inductor{} loop (FFmpeg \texttt{filter\_complex} \cite{ffmpeg}, the
  library surface evaluated in \autoref{sec:eval-invariance}). Anti-unification aligns sibling filter-graphs
  into a chain skeleton \cite{plotkin1970generalization,babble2023popl}; the
  induced fragment is mostly \emph{static} --- filter names, per-filter option
  sets, per-option enum domains, fixed once from the tool --- squeezed by the
  hard-gate between corpus positives and auto-generated negatives
  \cite{arvada2021ase,hygenar2025acl}; only pad labels stay dynamic, resolved
  against the live $\Gam$ at decode time (declare-then-use within one command,
  \autoref{sec:sdcg}). The gate validates the proposal against the same executable grammar
  object used online. Lineage: \autoref{sec:rel-axisD}.}
  \label{fig:induction-loop}
\end{figure*}

\paragraph{Where fragments come from: agentic induction under a hard gate.}
The construction above consumes a fragment library and a policy
(\autoref{fig:induction-loop}).
Both are induced offline, once per DSL, by
a strong model: programming-language-specific knowledge distillation
\cite{hinton2015distill}, except that the product is a grammar --- it edits
the support set, not weights, and is auditable and reusable across weak
models. The loop has four steps.
\begin{enumerate}\setlength{\itemsep}{1pt}
\item From $K$ same-family programs, anchor-based anti-unification
\cite{babble2023popl,plotkin1970generalization} extracts a skeleton and its
variation points.
\item The LLM decides which surface differences are the same semantic
parameter; on the GEMM corpus this recovers
\texttt{block\_M}/\texttt{N}/\texttt{K} as one tile family (\autoref{sec:evaluation}).
\item Each variation point becomes a fragment. Where the LLM observes
``always a previously declared name,'' the position becomes a $\Gam$-slot
plus an \texttt{extract} action: the loop has learned the context-dependence.
\item A hard gate auto-generates counter-examples and squeezes the proposed
GBNF between corpus positives and these negatives
\cite{hygenar2025acl,arvada2021ase}. The gate compiles from the same grammar
object as the executor mask; for a $\Gam$-slot it first instantiates
\texttt{\%slot\%} and then checks that ghost names are rejected.
\end{enumerate}
Once reference marking, query well-formedness, and environment fidelity satisfy the premises above, support-set soundness is independent of the proposer. The hard gate checks executability, corpus coverage, and mined counterexamples; these tests validate candidate fragments but do not replace the stated structural premises. The strong model raises the ceiling---which holes to merge and which rung to recover---not the guarantee floor.
The proposer need not be an LLM: expert-written fragments, or fragments reflected
from existing metadata (compiler headers, a schema, a tool's self-description;
\autoref{sec:evaluation}), enter through the same gate; LLM induction is the route where neither
exists.
Induction is fragment-sized: the loop synthesizes one region's productions at
a time, never a whole language's grammar --- the classically hard problem
\cite{arvada2021ase}.

On a reference hole the loop induces four rungs along $\refn$: \texttt{base}
(general CFG) $\xrightarrow{\wallop}$ \texttt{gamma} (the declared set: the
\wallop{} edge) $\to$ \texttt{ctx} (attribute-specialized, e.g.\
shape-filtered) $\to$ \texttt{pin} (gold), plus the fallback $\pol$ of \autoref{sec:sdcg}.
Each rung is scored on generation utility rather than parse accuracy, by two
static, weight-independent measures --- \textbf{coverage} (is the corpus
reconstructible, by template reverse-matching) and \textbf{\freebits}
(per-hole log-support trace score) --- plus \textbf{oracle-pass} (the rate an external
oracle accepts samples under the rung's mask).

\paragraph{The boundary.}
Proposition~2 delimits what the monotone $\Gam$-scope mechanism enforces
directly and what requires a richer analysis.

\begin{proposition}[Guarantee boundary of the $\Gam$-scope model]
Within the operator-skeleton and $\Gam$-scope model, the following property
classes are enforceable at the mask: \textbf{P1} scope safety, \textbf{P2}
sort/type-tag match via kind-filtered slots, \textbf{P3} declaration-before-use
at hole granularity, and \textbf{P4} shape agreement when shapes are represented
as write-once slot attributes. \textbf{N1} semantic-value correctness,
\textbf{N2} general termination, and \textbf{N3} invariants over mutable
attributes are not reducible, in general, to membership in this monotone
environment; enforcing them requires a semantic verifier, a stateful monitor,
or an externally verified support set.
\end{proposition}

N1 ranges over program semantics rather than $\Gam$ membership and is
undecidable in general \cite{rice1953}, although particular decidable properties
can be supplied by an external monitor. N2 contains the halting problem
\cite{turing1936}; N3 requires state beyond the monotone layer used by the
proof. At such holes the policy either invokes the richer analysis or falls
back to a coarser fragment. Appendix~A.5 gives the full table.

\paragraph{Necessity of dynamic reference support.}
\label{sec:prop3}
The grammar skeleton may be compiled ahead of time; its exact reference support
cannot be fixed ahead of the generated environment. Remark~1 rules out one
fixed CFG for declaration consistency; Proposition~3 rules out the broader
static alternative, a finite family of precompiled grammars with fixed
reference supports, under any dispatch function that selects among them by
prefix. Compilers separate context-free parsing from symbol-table-based name
resolution. Generation brings that resolution into the causal loop: each
reference position receives support computed from the prefix that precedes it.

A \emph{static masking system} consists of a finite grammar family
$\mathfrak{G}=\{G_1,\dots,G_m\}$ fixed before decoding, and a dispatch function
$\delta$ from prefixes to grammars. It satisfies the \emph{static-support
assumption} if the support presented at a reference position is drawn from a
finite reference-support library fixed before decoding. This captures ordinary
static grammar composition and tag dispatch. A system whose support range over
prefixes is not contained in any finite precompiled library already performs
the dynamic reference-support computation isolated here.

The system is \emph{sound} if no reference-position support contains an
undeclared name. It is \emph{non-blocking} if it never masks out a scope-safe
continuation.

\begin{proposition}[Necessity and exactness of dynamic reference support]
\label{prop:necessity}
Let the DSL contain a declare-then-use form over an identifier space
$\mathrm{Id}$ with a two-letter subspace. Then:

\smallskip
\noindent\textup{(i)}
no single CFG accepts exactly $L_{\mathrm{scope}}$, the
declaration-consistent language;

\smallskip
\noindent\textup{(ii)}
No static masking system under the static-support assumption is both sound and
non-blocking.

\smallskip
\noindent\textup{(iii)}
A masking system is sound and non-blocking at every reference position iff its
support at prefix $p$ is exactly $\dom(\Gam(p))$.

\smallskip
Therefore, over an unbounded identifier space, any sound, non-blocking mask
must synthesize from the prefix the exact reference support
$S(p)=\dom(\Gam(p))$. Decode-time instantiation $\Ginst{G}$ with the unfiltered
declared-name query realizes this necessary support directly; filtered queries
for sort or shape sit below it, remaining sound by refinement while
intentionally giving up name-level non-blocking for stronger guarantees.
\end{proposition}

\begin{proof}[Proof sketch]
Part (i) is Remark~1. For (ii), prefixes can produce infinitely many distinct
declared sets $D$, while a static system has only finitely many reference
supports. Soundness and non-blocking force the support at each prefix to be
exactly $D$: soundness gives $S\subseteq D$, and non-blocking gives
$D\subseteq S$. By pigeonhole, two distinct declared sets must share a support,
contradicting exactness. Part (iii) is exactly the two inclusions just stated.
Our decode-time instantiation realizes this support by querying
$\dom(\Gam(p))$;
Theorem~1 gives soundness, and exact instantiation gives non-blocking for the
unfiltered query. Sound and non-blocking are the mask-level analogues of
preservation and progress in syntactic type soundness
\cite{wrightfelleisen1994}. Full proof: Appendix~A.4.
\end{proof}

Clause (ii) composes Floyd's classical observation with a finite-support
pigeonhole argument. Clause (iii) is the architectural characterization:
independent of implementation, the maximal sound reference support is exactly
the declared set computed from the prefix. Our self-extending grammar realizes
this necessary runtime interface by materializing that set into the mask.

If the identifier space were artificially bounded to $n$ names, a static family
of size $2^n$ could pre-enumerate all supports. Real DSLs have unbounded or
effectively unbounded name spaces, so no finite fixed-support family suffices.

\subsection{Clients: Induced Grammars in the Wild}
\label{sec:clients}

The same \texttt{reference\_ladder} machine instantiates all three client
surfaces of \autoref{sec:evaluation}; only the anchor sort and the source of $\Gam$ change. We show
the induced fragments as they exist in the artifact. Each passed the hard-gate
of \autoref{sec:construction} (positives accepted, mined ghosts rejected) before use.

\paragraph{Tensor-kernel operands (TileLang).}
Each induced fragment exists in two states. What induction produces is
\emph{static}: the carrier production with its operand positions as slot
\emph{expressions}, each query attached but unevaluated --- at induction time
there is no runtime $\Gam$ to evaluate against:

\begin{lstlisting}[language={}]
gemm ::= "T.gemm(" %op% ", " %op% ", " %acc% ")"   -- induced fragment: slots unevaluated
  %op%  <- Gamma.names(sort=Buffer, mem=shared)    -- attached slot queries
  %acc% <- Gamma.names(sort=Buffer, mem=fragment)
\end{lstlisting}

\noindent At decode time the executor evaluates the queries against the live
environment and splices the results (\autoref{sec:sdcg}). At a $\Gam$ declaring
\texttt{A\_shared}, \texttt{B\_shared} (shared) and \texttt{C\_local}
(accumulator fragment), the same fragment becomes:

\begin{lstlisting}[language={}]
gemm ::= "T.gemm(" ("A_shared" | "B_shared") ", " ("A_shared" | "B_shared") ", " "C_local" ")"
\end{lstlisting}

\noindent The support is thus \textbf{two-layered}. The static layer ---
skeleton, slot positions, queries --- is what induction can derive with no
runtime state; the dynamic layer of query evaluation, splicing, and masking is
what the executor supplies at decode time. The role filter (shared operands
feed the A/B positions, the fragment buffer the accumulator) is the
\texttt{ctx} rung's attribute specialization; dropping it recovers the
\texttt{gamma} rung, whose single slot is the whole declared set.

\paragraph{Schema columns (SQL)}
The column hole's ladder queries $\Gam$ populated by schema introspection; the
\texttt{ctx} rung (\texttt{gamma\_table} in \autoref{sec:eval-sql}) intersects with the
\texttt{FROM} table's columns:

\begin{lstlisting}[language={}]
sel ::= "SELECT " col " FROM singer" ...
col ::= "name" | "country" | "age"   -- Gamma.names(sort=Column) restricted to columns(singer)
\end{lstlisting}

\noindent A cross-table column is legal at the \texttt{gamma} rung but excluded
at \texttt{ctx}; a ghost column is excluded at both.

\paragraph{Live tool state (CLI)}
For \texttt{git} \cite{gitscm}, $\Gam$ holds the live refs --- and grows as commands execute
(a branch is referenceable only after \texttt{checkout -b} creates it):

\begin{lstlisting}[language={}]
checkout ::= "git checkout " ref
ref ::= "main" | "dev" | "release-2.0" | "hotfix-login"   -- Gamma.names(sort=GitRef): live branches
\end{lstlisting}

\noindent The same ladder runs a third $\Gam$ source; no grammar was written
by hand for any of the three. The inductor works fragment by fragment, squeezing
each executable grammar between corpus positives and generated negatives before
the online executor uses it. \autoref{sec:evaluation} measures what these fragments buy at decode
time.

\section{Implementation and Scope}
\label{sec:method}

The two components of \autoref{sec:intro} realize the two halves of \autoref{sec:formalization}. The
\textbf{executor} (\gproj) runs online, decoding each region-hole under the
fragment selected by $\pol(\sort,\Gam)$ on the grammar-call stack. The
\textbf{inductor} (\inductor) runs offline, constructing the four-rung ladder,
reference markings, $\Gam$-slots, extract actions, and fallback policy under the
hard gate. At runtime, the executor JIT-instantiates this artifact at the typed
holes of the model's sketch, so the resulting analysis facts act directly on
token support.

\paragraph{One invariant, kept two ways.}
At each constrained reference position, the instantiated support is exactly
$\queryslot(\Gam(p))$, with
$\queryslot(\Gam(p))\subseteq\dom(\Gam(p))$. The \textbf{executor} enforces this
invariant online by re-instantiating slots per hole and recording generated
declarations in $\Gam$; the \textbf{inductor} constructs the reference
markings, queries, and fragment policy offline. The shared hard gate checks
executability, corpus coverage, and mined ghosts. At name-complete positions,
Proposition~3 characterizes $\dom(\Gam(p))$ as the maximal sound,
non-blocking support and rules out finite fixed-support precompilation.
Accordingly, the executor JIT-specializes grammar fragments from the evolving
environment, while the inductor prepares their static structure.

\subsection{Two-stage constrained synthesis}
\label{sec:method-two-stage}

The workflow separates \textbf{Stage~1---sketch search}, which uses prompting,
multiple attempts, or self-repair to explore semantic structure, from
\textbf{Stage~2---constrained realization}, which fills reference-sensitive
details in one masked pass under $\pol(\sort,\Gam)$. The inductor supplies the
policy and the executor performs Stage~2; \autoref{sec:eval-roadmap} evaluates the
pipeline end to end on a production kernel.

\subsection{Backends and engineering}
\label{sec:method-backends}

Masking and \freebits{} logic are backend-agnostic: the same executor drives
in-process HF+XGrammar and vLLM serving, and a $\Gam$ fixed before decoding
(SQL) and one growing during it (TileLang) share one code path.
Each hole has its own token budget
(\texttt{max\_new\_tokens}), and grammar-call nesting is capped at
\texttt{\_MAX\_NEST\_DEPTH = 32}. Unit and GPU end-to-end tests include
32-deep nesting, 300-name $\Gam$ sets, and maliciously injected names.

\subsection{Applicability}
\label{sec:method-scope}

Every directly supported property has the form
\textbf{skeleton + reference filling}: a CFG-expressible skeleton whose
reference positions select entities from a runtime $\Gam$ (the P1--P4 side of
Proposition~2). Properties requiring execution, whole-program relations, or
mutable fixed-point state compose through a verifier or monitor (N1--N3), while
the grammar retains its construction-level guarantee.

The strongest fits are reference-dense, declaration-scoped, low-resource
surfaces: tensor-kernel DSLs such as TileLang and Triton, query languages such
as SQL and jq, and P4 (\autoref{sec:evaluation}). The same abstraction covers library/API usage by
placing available functions, intrinsics, options, and version-specific names in
$\Gam$ (\autoref{sec:eval-invariance}--\autoref{sec:eval-induction}). General-purpose
C++ and Python often require heap, alias, synchronization, or scheduling facts
beyond name membership, while regex has no runtime environment to specialize.
\autoref{sec:conclusion} discusses extending the same decode-loop interface to richer analyses.

\subsection{Limitations}
\label{sec:method-limitations}

\textbf{Guarantee boundary.} No-Ghost and P1--P4 hold under their stated
premises; compilation and functional correctness remain oracle properties, and
a budget-exhausted fragment emits no completed surface. \textbf{Engine
correspondence.} Theorem~1 is modulo Assumption~1, validated on the rendered-slot
slice by differential testing and token-mask replay. \textbf{Composition.} The
$\mathit{Times}$ combinator does not yet compose with nested grammar calls;
refinement is a preorder, and cross-boundary sort agreement additionally
requires a correspondence between host and embedded sort vocabularies.

\section{Evaluation}
\label{sec:evaluation}

We evaluate four questions: \textbf{(Q1)} does replacing an open reference with a $\Gam$-typed slot---the $\wallop$ edge---eliminate ghost references by construction? \textbf{(Q2)} does this guarantee transfer across models, programming paradigms, and sources of $\Gam$? \textbf{(Q3)} what does it buy relative to prompting, retry, and tighter fragments? \textbf{(Q4)} what serving overhead does runtime specialization add? We first isolate the $\wallop$ cliff (\autoref{sec:eval-wall}), then test its transfer across settings (\autoref{sec:eval-invariance}), compare it with cheaper biasing on SQL (\autoref{sec:eval-sql}), evaluate offline induction (\autoref{sec:eval-induction}), compose the pieces end to end (\autoref{sec:eval-roadmap}), and finally quantify serving overhead (\autoref{sec:eval-overhead}).

\paragraph{Setup.}
We use Qwen3-0.6B as the primary model and replicate selected results on gpt-oss-20B, Qwen3.5-27B-distill, a 120B Nemotron served by vLLM, and the 236B DeepSeek-V4-Flash \cite{qwen32025arxiv,vllm2023sosp}. Oracles are the actual target systems: TileLang compilation plus a torch numerical reference \cite{tilelang2025arxiv}, SQL execution-match on Spider \cite{testsuite2020emnlp,spider2018emnlp}, the \texttt{p4c-bm2-ss} compiler, library-specific validators, and the command-line tools themselves. Unless stated otherwise, each controlled $\wallop$ ablation holds the model, prompt, sampling configuration, and surrounding fragment fixed, changing only the source of a reference: an open identifier (\texttt{free}) versus a $\Gam$-typed slot (\texttt{gamma}) or its context-filtered refinement. We report support freedom using \freebits{} and distinguish the construction-level No-Ghost guarantee (Theorem~1, modulo Assumption~1) from empirical oracle pass rates.

\subsection{The cliff, and its cause}
\label{sec:eval-wall}

Scope safety changes at exactly the $\wallop$ edge. On one TileLang GEMM task with Qwen3-0.6B, we progressively loosen only the fragment used for operand holes. Across four rungs, \freebits{} rises from $92$ to $287$, $538$, and $998$, while compilation is $100\%$, $100\%$, $67\%$, and $0\%$. The final transition changes only operand references from $\Gam$-typed candidates to open identifiers; every resulting failure is a ghost-buffer \texttt{NameError}. The $67\%$ is 2-of-3 seeds in this diagnostic sweep; its residual failure is an unrelated warp-count error, while the $n=12$ rerun reported in \autoref{tab:wall-summary} reaches $75\%$. Thus increasing freedom does not itself cause the collapse: compilation falls to zero precisely when references leave $\Gam$. \autoref{tab:wall-summary} tests the same edge across models and surfaces.

\subsection{Invariance: capacity, paradigm, surface}
\label{sec:eval-invariance}

Theorem~1 makes No-Ghost independent of model weights and programming surface once its premises hold. We therefore test whether the empirical failure motivating the guarantee recurs across models, paradigms, and sources of $\Gam$. \autoref{tab:wall-summary} summarizes the controlled \texttt{free} versus $\Gam$-typed ablations: every $\Gam$-typed arm has zero ghosts by construction, whereas every \texttt{free} arm fails through ghost references.

\paragraph{Capacity.}
With operand references open, TileLang compilation is $0\%$ for the 0.6B, 20B, 27B, and 120B models. All 36 HF-model failures and both 120B endpoint failures are ghost-name errors. We make no capacity claim from the intermediate $\Gam$-typed compilation rates, whose confidence intervals overlap and whose residual failures are unrelated to scope; the relevant result is that increased capacity does not remove off-$\Gam$ references.

At the frontier, unconstrained DeepSeek-V4-Flash transfers a neighboring TVMScript dialect into TileLang: all eight initial samples use incompatible API forms, and all four rigorous toolchain reruns fail on undeclared references. Under induced TileLang grammars, the same model produces real-API, ghost-free programs on 12 GEMM and 8 elementwise samples. All GEMM samples and seven elementwise samples also pass the numerical oracle; the remaining elementwise sample compiles but exceeds shared-memory capacity. Stronger models improve the surrounding program, but the reference guarantee still comes from the support constraint.
\label{sec:eval-strongmodel}

\paragraph{Paradigm.}
The same edge appears outside TileLang. On Spider SQL, open column identifiers produce ghost columns in all 90 samples and $0\%$ execution-match, whereas restricting the column to the selected table's schema gives $100\%$ execution-match (\autoref{sec:eval-sql}). On P4, open references yield $0\%$ compilation with 8/8 ghost failures, while the multi-sort $\Gam$ arm compiles all samples with zero ghosts. The result replicates with gpt-oss-20B on \texttt{basic.p4} and with both the 0.6B and 20B models on the structurally distinct \texttt{l2fwd.p4}; every $\Gam$-arm sample compiles with zero ghosts. These cases cover imperative, declarative, and data-plane languages with the same declaration-and-reference mechanism.

\paragraph{Surface.}
The mechanism also applies beyond DSLs. For \texttt{git checkout}, selecting from live branches passes the repository oracle in all eight samples, whereas an open identifier produces eight plausible but nonexistent branch names. The pattern repeats for \texttt{make} targets and \texttt{npm} scripts. Static API surfaces exhibit the same failure: blind AVX-512 generation invents eight nonexistent intrinsics, all excluded by a fragment reflected from the compiler headers; FFmpeg generation similarly produces ghost options, version-specific names, and undeclared pad labels. These masks guarantee availability, not intent: the model may still select the wrong legal branch, option, or operand.

\begin{table}[tb]
\centering
\small
\begin{tabular}{lllc}
\toprule
Paradigm / surface (oracle) & $\Gam$ pass\% & free pass\% & ghost (free) \\
\midrule
TileLang compile, 0.6B & 75\% & 0\% & 12/12 \\
\quad 20B / 27B & 83\% / 67\% & 0\% / 0\% & 24/24 \\
SQL exec-match,\textsuperscript{$\dagger$} 0.6B & 100\% & 0\% & 90/90 \\
P4 compile, 0.6B & 100\% & 0\% & 8/8 \\
CLI \texttt{git} rev-parse, 0.6B & 100\% & 0\% & 8/8 \\
\bottomrule
\end{tabular}
\caption{The $\wallop$ cliff across models and programming surfaces. Every $\Gam$-typed arm is ghost-free by construction; residual oracle failures are unrelated to scope. Every \texttt{free} arm shown fails through ghost references. TileLang uses $n=12$ per model, SQL $n=90$, and P4/CLI $n=8$.}
\label{tab:wall-summary}
\end{table}

\subsection{Why by construction: functional correctness vs.\ cheaper biasing}
\label{sec:eval-sql}

By-construction scope safety is useful only if it contributes to functional correctness under a real oracle and offers something prompting or retry does not. We use SQL/Spider execution-match, which compares result rows rather than merely checking whether a query runs.

\paragraph{The oracle is stronger than executability.}
On a committed hard subset of 30 Spider tasks with three seeds each, unconstrained generation executes in $87\%$ of cases but matches the reference result in only $76\%$, confirming that execution-match tests more than syntactic or runtime validity. The context-filtered $\Gam$ arm reaches $100\%$ execution-match while remaining ghost-free by construction. On a separate simple-query subset, the refinement ladder shows the expected trade-off: execution-match rises from $0\%$ at \texttt{free} to $85\%$ at \texttt{gamma} and $100\%$ at \texttt{gamma\_table} and \texttt{fine}, while \freebits{} falls from $249.8$ to $17.9$, $13.8$, and $0$. Replication on the 20B and 27B models again gives $0\%$ execution for open column identifiers, while the context-filtered arms remain ghost-free. The guarantee remains scope safety rather than full SQL correctness: a legal column may still be assembled into the wrong query.

\paragraph{Cheaper alternatives help but do not guarantee scope.}
On the same 60 Spider samples, every prompt contains the schema; the \texttt{free} and context-filtered masked arms share the same literal prompt and differ only in the mask. Explicitly listing the selected table's legal columns raises unconstrained execution-match to $93.3\%$, but still produces four ghost columns. Retrying execution failures up to four times reaches $90.0\%$, produces five ghosts, and uses 40.0 generated tokens on average. The context-filtered mask reaches $100\%$ in one pass with no ghost columns and 6.6 generated tokens on average. Prompting and retry improve the distribution; only masking removes invalid columns from its support.

\begin{table}[tb]
\centering
\small
\begin{tabular}{lrrrc}
\toprule
Method & exec-match\% & ghost cols & avg tokens & by-constr.\ ghost-free? \\
\midrule
free (open ident) & 0.0 & 60 & 20.8 & --- \\
unconstrained & 70.0 & 11 & 22.9 & no \\
prompt-schema & 93.3 & 4 & 17.2 & no \\
retry ($k\le 4$) & 90.0 & 5 & 40.0 & no \\
\textbf{gamma\_table (\wallop)} & \textbf{100} & \textbf{0} & \textbf{6.6} & \textbf{yes} \\
\bottomrule
\end{tabular}
\caption{The mask versus cheaper alternatives on the same 60 Spider samples, oracle, and 0.6B model.}
\label{tab:sql-cheaper-baselines}
\end{table}

\paragraph{When feedback is absent.}
CMake exposes a stricter limitation of repair: 8 of 11 observed error classes produce no diagnostic. Misspelled variables silently expand to empty strings, and version-sensitive predicates may simply evaluate false. Retry and self-repair therefore receive no failure signal, whereas $\Gam$-typed fragments exclude unknown variables and enumerable compiler properties before generation.

\subsection{Installing it offline: inducing the configuration}
\label{sec:eval-induction}

Offline induction is evaluated without model decoding: corpus positives must be accepted and mined ghost substitutions rejected by the same parser object later used for masking. This isolates whether the inductor can construct executable $\Gam$-typed fragments from a real low-resource surface.

\paragraph{AscendC / \texttt{ops-nn}.}
We induce fragments from Huawei CANN's \texttt{ops-nn} library, covering 22 operator categories and hundreds of operators. One fragment family covers 69 sibling \texttt{foreach} operators sharing a kernel skeleton; another captures the role-typed \texttt{mm\_} protocol across two matrix-multiplication operators. The induced fragments recover three forms of context dependence beyond ordinary declaration-before-use: a variant choice that determines downstream fragment shape, positional role filters that reduce a slot to the matching argument, and registered combinations that form a strict subset of the Cartesian product of their axes. Every fragment passes the shared hard gate on corpus examples and mined ghosts. Because the available environment lacks a local CANN compiler, this case evaluates family coverage and parser-level support construction rather than end-to-end compilation.

The same inductor recovers the four-rung ladder (\texttt{base}$\to$\texttt{gamma}$\to$\texttt{ctx}$\to$\texttt{pin}) on SQL, jq, and TileLang. On jq, the \texttt{base} rung returns silent nulls for ghost fields, while \texttt{gamma} answers all seven cases under the real oracle; Appendix~A.7 reports coverage and monotone \freebits{} along the induced ladders.

\paragraph{A private library isolates negative transfer.}
On a table API designed to differ systematically from pandas \cite{privatelibrary2022emnlp}, full documentation yields 18/18 passing generations, whereas names-only and no-documentation conditions produce parameter-level and API-level negative transfer. Across 15 observed error classes, the induced \texttt{gamma} and \texttt{ctx} fragments exclude 13 at the support level; the remaining failures require value semantics or a stateful call-order protocol. This case separates reference and parameter constraints from properties outside the monotone $\Gam$ model.

\subsection{Composition and the end-to-end pipeline}
\label{sec:eval-roadmap}

The full two-stage pipeline is evaluated on grouped-query attention (GQA) \cite{gqa2023emnlp}: a sketch model selects a policy under a constrained enumeration, and an implementation model decodes the kernel body under the resulting masks. With Qwen3-0.6B, the $\Gam$-typed arm compiles 8/8 samples, while both open-identifier and unconstrained arms compile 0/8. The $\Gam$ arm remains ghost-free across the 0.6B, 20B, and 27B models.

Three extensions compose without changing the masking engine. First, adding shape attributes to the GQA operand query reduces wrong-shape selection from $2.5\%$ under kind-only filtering to zero by construction. Second, nested grammar calls produce CUDA-compiling, numerically matching kernels in both evaluated seeds. Third, a JSON-schema composition with declare-then-use references is reference-consistent by construction, a dependency not expressible in plain JSON Schema \cite{synchromesh2022iclr}. Appendix~A.6 gives the complete protocols.

\subsection{Serving overhead}
\label{sec:eval-overhead}

\autoref{tab:serving-overhead} compares the serving overhead of unconstrained decoding, XGrammar, and the full \gproj{} framework via token throughput. All measurements is conducted on a server with 2 Intel Xeon Gold 6430 CPUs and 512 GB memory, along with 4 RTX PRO 6000 GPUs. We use vLLM 0.24.0 and XGrammar 0.2.3 as the software backbones. DeepSeek-V4-Flash is served in FP8 with tensor parallelism $\mathrm{TP}=4$; Qwen3-0.6B is served on one RTX PRO 6000. We report aggregate generation throughput at batch sizes 1 and 32.

\begin{table}[tb]
\centering
\small
\begin{tabular}{llrrrr}
\toprule
Model & Workload & Batch & Unconstr. & XGrammar & \gproj{} \\
\midrule
DeepSeek-V4-Flash & SQL      & 1  & 5.221    & 5.217    & 4.664 \\
DeepSeek-V4-Flash & SQL      & 32 & 161.339  & 160.195  & 142.093 \\
DeepSeek-V4-Flash & TileLang & 1  & 5.212    & 5.211    & 4.648 \\
DeepSeek-V4-Flash & TileLang & 32 & 159.418  & 158.321  & 140.114 \\
Qwen3-0.6B        & SQL      & 1  & 87.963   & 86.048   & 71.936 \\
Qwen3-0.6B        & SQL      & 32 & 2374.071 & 2008.651 & 1663.163 \\
Qwen3-0.6B        & TileLang & 1  & 87.927   & 86.102   & 70.776 \\
Qwen3-0.6B        & TileLang & 32 & 2370.325 & 2126.275 & 1771.187 \\
\bottomrule
\end{tabular}
\caption{Serving throughput in generated tokens per second (higher is better). }
\label{tab:serving-overhead}
\end{table}

Across the eight configurations, \gproj{} reduces throughput by 10.6--17.8\% relative to XGrammar (14.0\% on average). Relative to unconstrained decoding, the reduction is 17.3\% on average: 10.7--12.1\% for DeepSeek-V4-Flash and 18.2--29.9\% for Qwen3-0.6B. The larger Qwen batch-32 end-to-end gap includes XGrammar's own 10.3--15.4\% reduction; the incremental cost of \gproj{} in those settings remains 16.7--17.2\%. Thus runtime specialization adds a low-teens overhead on top of the masking backend rather than a severe serving regression.
\FloatBarrier

\section{Related Work}
\label{sec:related-work}

Constrained code generation can be organized along four questions: \textbf{(A)} whether a method changes model weights or the decoder's support, \textbf{(B)} what object it constrains, \textbf{(C)} when and at what granularity the constraint is selected, and \textbf{(D)} where the constraint comes from. Our position is a support-set method whose object is a refinement-ordered family of grammar fragments, selected per region from the live environment and induced offline when no specification is available.

\begin{table}[t]
\caption{Capability boundary among the closest policy-level approaches. \checkmark{} denotes a first-class capability, $\triangle$ a specialized or partial form, and -- absence. ``Fresh'' means semantic support synthesized from the generated prefix or a live analysis, rather than dispatch over fixed supports. CSD establishes prefix-derived completion; type-constrained decoding and ChopChop provide deeper property-specific guarantees; ToP provides context-populated modular parsers; and LMQL/guidance provide imperative regional orchestration. Our claim is the complete combination, not dominance in semantic expressiveness.}
\label{tab:capability-boundary}
\centering
\scriptsize
\setlength{\tabcolsep}{2.6pt}
\renewcommand{\arraystretch}{1.04}
\begin{tabularx}{\textwidth}{@{}>{\raggedright\arraybackslash}p{0.205\textwidth}ccccccX@{}}
\toprule
& \multicolumn{2}{c}{Dynamic support} & \multicolumn{3}{c}{Policy object} & & \\
\cmidrule(lr){2-3}\cmidrule(lr){4-6}
System & live & fresh & region & order & fallback & induced & Established result \\
\midrule
Synchromesh/CSD \cite{synchromesh2022iclr} & \checkmark & \checkmark & $\triangle$ & -- & -- & -- & CE-valid completion \\
Type-constrained \cite{typeconstrained2025pldi} & \checkmark & \checkmark & $\triangle$ & -- & -- & -- & type soundness \\
Tree-of-Parsers \cite{top2025colm} & \checkmark & \checkmark & \checkmark & -- & -- & -- & semantic/runtime guarantees \\
ChopChop \cite{chopchop2026popl} & \checkmark & \checkmark & $\triangle$ & -- & -- & -- & soundness/completeness \\
LMQL/guidance \cite{lmql2023pldi} & $\triangle$ & $\triangle$ & \checkmark & -- & $\triangle$ & -- & imperative orchestration \\
Grammar Prompting \cite{grammarprompting2023neurips} & -- & -- & -- & $\triangle$ & -- & $\triangle$ & specialized BNF \\
\textbf{Decode-Time Grammars} & \checkmark & \checkmark & \checkmark & \checkmark & \checkmark & \checkmark & Thm.~1, Lem.~2, Prop.~3 \\
\bottomrule
\end{tabularx}
\end{table}

\subsection{Axis A --- weights or support}
\label{sec:rel-axisA}

Weight-axis methods make desired continuations more probable through training, prompting, exemplars, repair, or probabilistic control. Grammar-aligned decoding adjusts model probabilities under a grammar constraint \cite{grammaraligneddecoding2024neurips}, while sequential-Monte-Carlo methods weight and resample candidates according to semantic potentials \cite{smc-llm-control2025iclr}. For low-resource languages, MultiPL-T translates high-resource programs into test-validated semi-synthetic corpora and fine-tunes pretrained code models on the result \cite{multipl-t2024oopsla}. This data-and-weights route is complementary to ours: we leave model weights fixed and induce reusable support constraints from corpus and schema evidence. These methods can substantially improve generation, but they do not generally remove an invalid continuation from the support before it is sampled.

Support-axis methods instead make selected continuations impossible. XGrammar, SynCode, and Outlines efficiently implement regular and context-free masks \cite{xgrammar2025mlsys,syncode2025tmlr,outlines2023arxiv}; budget-aware masking further excludes prefixes that cannot reach acceptance within the remaining token budget \cite{dang2026mitigating}. Our guarantees live on this axis: once a fragment and runtime environment satisfy the premises of Theorem~1, reference safety does not depend on model capacity or preference.

\subsection{Axis B --- the object constrained}
\label{sec:rel-axisB}

The standard object is a single grammar, regex, or JSON schema used as a recognizer. Richer systems constrain semantic properties. ChopChop formulates semantic constrained decoding through realizability and proves program-level soundness and completeness properties \cite{chopchop2026popl}; type-constrained generation lowers type reachability into the decoder while deliberately trading completeness for tractability \cite{typeconstrained2025pldi}. These analyses are complementary to our refinement order: a semantic monitor can be attached to a fragment selected at any rung.

Our object is a policy over multiple fragments for the same semantic construct. Tree-of-Parsers (ToP) is the closest executor mechanism: it uses modular grammars with runtime-populated slots and sub-parsers \cite{top2025colm}. ToP organizes generation around a fixed construct-to-template structure; we add a refinement-ordered family, per-region rung selection, fallback, and offline induction. Krogmeier and Madhusudan provide the closest conceptual precedent for reasoning over multiple grammars for one base language and selecting among them by learning utility \cite{krogmeier2023oopsla}; we carry this perspective into decode time, where the selected fragment is specialized from the live $\Gam$.

The connection between grammars and environments is classical. Attribute grammars, two-level grammars, definite-clause grammars, adaptable grammars, and data-dependent grammars attach contextual information to syntax during recognition \cite{knuth1968attribute,vanwijngaarden1975algol68,dcg1980,christiansen1990adaptable,shutt1993rag,yakker2010popl}. Floyd's ALGOL result already shows why declaration consistency cannot be captured by an ordinary fixed CFG \cite{floyd1962algol}. Our distinction is operational: instead of analyzing a completed input, we specialize an environment-indexed fragment before emitting the corresponding region and compile the resulting facts directly into token support.

\subsection{Axis C --- when and how finely it is chosen}
\label{sec:rel-axisC}

Some systems select a constraint once before decoding. Grammar Prompting predicts an instance-specific subgrammar and then decodes under it \cite{grammarprompting2023neurips}. Other systems expose mid-decode constraints: LMQL and guidance provide imperative per-slot orchestration \cite{lmql2023pldi}, and ToP invokes context-populated sub-parsers at designated template positions. These mechanisms demonstrate the value of intermediate constraints. Their structure is manually specified or organized around fixed template positions; our policy instead makes selection among a refinement family explicit while specializing each selected fragment from the live environment.

Semantic analyses can also depend on the generated prefix. Synchromesh lowers language-specific completion information into admissible continuations \cite{synchromesh2022iclr}. Monitor-Guided Decoding provides a complementary analyzer-backed route, querying repository-aware static analysis at predefined trigger points and lowering type-consistent suggestions into token masks \cite{mgd2023neurips}; IterGen handles SQL-column dependencies through generate--check--rollback \cite{itergen2025iclr}. These systems establish that generated context can guide subsequent decoding. Our policy makes the \emph{configuration} of that guidance a first-class object: it is selected per region as $\pol(\sort,\Gam)$, instantiated before that region is emitted, and can fall back through the refinement order. At name-complete reference positions, Proposition~3 shows that exact sound and non-blocking support cannot be represented by any finite family of fixed reference supports; it must be synthesized from the prefix. This result distinguishes runtime support construction from static grammar or tag dispatch, rather than from dynamic semantic constrained decoding in general.

\subsection{Axis D --- where the grammar comes from}
\label{sec:rel-axisD}

Most deployed constraints are hand-written. Grammar-induction systems instead approximate a target language from examples: ARVADA synthesizes grammars from program inputs, and HyGenar combines LLM proposals with search and surface-based validation \cite{arvada2021ase,hygenar2025acl}. Babble uses anti-unification to derive reusable abstractions for symbolic search, providing the lineage for our anchor-based alignment \cite{babble2023popl,plotkin1970generalization}.

Runtime context has also conditioned learned generators directly. PHOG conditions a probabilistic code model on program context \cite{phog2016icml}, and Krishnamurthy et al.\ constrain neural semantic parsing using table-specific runtime types \cite{krishnamurthy2017emnlp}. Our carrier is different: the model remains external and frozen, while an induced fragment library and policy alter its support per region. The inductor targets generation utility rather than whole-language recognition accuracy, using corpus coverage, \freebits{}, oracle pass rate, and a hard gate over executable fragments and mined counterexamples (\autoref{sec:construction}).

\subsection{Infrastructure and inherited foundations}
\label{sec:rel-infra}

We build on XGrammar's fast grammar compilation \cite{xgrammar2025mlsys}. XGrammar-2 adds efficient JIT compilation, cross-grammar reuse, and tag-triggered structural dispatch \cite{xgrammar22026cais}. These are complementary execution mechanisms: \gproj{} adds the grammar-call stack, per-hole runtime-$\Gam$ instantiation, declaration extraction, and lexical frame overlays that turn an evolving program environment into fresh reference support. Thus XGrammar supplies the masking backend, while \gproj{} supplies the environment-dependent orchestration.

We also inherit the standard abstraction of holed terms. Hazel studies typed holes during interactive editing \cite{hazel2017popl}, and Sketch fills holes through solver search \cite{sketch2006asplos}; here an LLM fills typed holes under fragments selected and specialized by the decode-time policy.

\subsection{Where we sit}
\label{sec:rel-position}

\autoref{tab:capability-boundary} makes the boundary explicit. Prior work establishes every ingredient separately, including dynamic semantic masks, programmable semantic realizability, runtime-populated parsers, learned specialized grammars, and efficient dynamic execution. Decode-time grammars are the first to organize these capabilities as an \emph{inducible refinement-ordered policy}: fragments are selected per region, specialized from an evolving $\Gam$, and weakened by explicit fallback. Theorem~1 establishes the resulting support-set reference guarantee, Lemma~2 preserves it under refinement, and Proposition~3 characterizes why exact name support at name-complete positions is necessarily a decode-time object. In this sense, constrained decoding becomes an execution substrate for incremental PL analyses whose results participate directly in constructing the generated program.

\section{Conclusion}
\label{sec:conclusion}

Decode-time grammars bind the decoder's support to the live runtime environment
$\Gam$. \gproj{} excludes ghost buffers, columns, APIs, options, and tool
references before sampling, while leaving semantic and algorithmic choices to
the model. Across DSLs, libraries, command-line tools, and models from 0.6B to
236B parameters, the same division of labor holds: stronger models improve the
program sketch, while the grammar provides a stable environment-bound reference
guarantee. Grammars need not only recognize syntax; during generation, they can
bind a program to the environment in which it must run.

More broadly, the decode loop can serve as an enforcement point for incremental
program analysis: an analysis derives facts from the evolving prefix, a support
transformer removes continuations those facts rule out, and a refinement policy
provides fallback when stronger constraints are unavailable. This paper realizes
that pattern for symbol-table scope. Extending it to richer prefix-indexed
analyses and learned switching policies is a natural next step.

\bibliographystyle{ACM-Reference-Format}
\bibliography{refs}
\ifincludeappendix
\appendix
\section{Full Proofs and Formal Details}
\label{app:full-proofs}

This supplement expands material the main text states and sketches: \autoref{app:gbnf}--\autoref{app:escaping} the
minimal-GBNF semantics and the Escaping Lemma behind Theorem~1 (standard grammar
substitution \cite{hopcroftullman1979,aho2006compilers}, with the one GBNF-specific
parenthesization subtlety); \autoref{app:engine-assumption} the
differential-testing protocol that discharges Assumption~1; \autoref{app:prop3-proof} the full proof of the
necessity theorem (Prop~3 of the main text) and its bounded-name reading; \autoref{app:boundary} the
$\Gam$-scope guarantee-boundary table in full; \autoref{app:composition} the lifting/composition details of the
evaluation; \autoref{app:refinement} the witnessing argument for the refinement order on the induction
ladder; \autoref{app:non-cf} the full non-context-freeness argument behind Remark~1; \autoref{app:p4-boundary} the
precise write-once restriction behind P4 (continuing the \autoref{app:boundary} table); \autoref{app:whole-program-proof} the full proof of the whole-program
corollary. Numbering of theorems here is local to this supplement; named references
(Theorem~1, Prop~3, \autoref{sec:formalization}, \autoref{sec:evaluation}) point into the main text.

\subsection{Minimal GBNF fragment and denotation}
\label{app:gbnf}
\paragraph{Minimal GBNF fragment.} Take a sub-grammar of GBNF sufficient for this lemma. The terminal alphabet is the set of Unicode scalar values (the engine consumes UTF-8); the metacharacters are $\{\,\texttt{"}\ \mid\ (\ )\,\}$. Grammar expressions:
$$
e \;::=\; \mathit{lit} \ \mid\ e\,e \ \mid\ e \mid e \ \mid\ (\,e\,) \ \mid\ N,
\qquad
\mathit{lit} \;::=\; \texttt{"}\,\mathit{ch}^*\,\texttt{"}
$$
where $\mathit{lit}$ is a \textbf{double-quoted string literal}, each $\mathit{ch}$ a verbatim character outside $\{\texttt{"},\backslash,\text{NL},\text{CR},\text{TAB}\}$ or one of $\backslash\backslash,\ \backslash\texttt{"},\ \backslash\mathtt{n},\ \backslash\mathtt{r},\ \backslash\mathtt{t}$ --- exactly the escape routine's output --- and $N$ a nonterminal. \textbf{Fixed disambiguation rules:} (i) concatenation $e\,e$ binds \textbf{tighter} than alternation $e\mid e$; (ii) parentheses $(\,e\,)$ group only; (iii) a string literal is a \textbf{single indivisible token} whose boundaries are its paired unescaped double-quotes.

\paragraph{Denotation} $\mathcal{L}\proj{\cdot} : e \to \mathcal{P}(\text{strings})$:
$$
\begin{aligned}
\mathcal{L}\proj{\texttt{"}s\texttt{"}} &= \{\,\mathrm{unescape}(s)\,\} &&\text{(literal: a one-string language)}\\
\mathcal{L}\proj{e_1\,e_2} &= \mathcal{L}\proj{e_1} \cdot \mathcal{L}\proj{e_2} &&\text{(concatenation = product)}\\
\mathcal{L}\proj{e_1 \mid e_2} &= \mathcal{L}\proj{e_1} \cup \mathcal{L}\proj{e_2} &&\text{(alternation = union)}\\
\mathcal{L}\proj{(\,e\,)} &= \mathcal{L}\proj{e} &&\text{(parentheses: identity)}
\end{aligned}
$$
with $\mathrm{unescape}\circ\mathrm{escape}=\mathrm{id}$. For nonterminals, fix a production environment $P$ mapping each nonterminal to its defining expression; the clause $\mathcal{L}\proj{N}=\mathcal{L}\proj{P(N)}$, read as the least solution of the resulting equation system (the standard least-fixed-point semantics of CFGs \cite{hopcroftullman1979}), makes $\mathcal{L}\proj{\cdot}$ total --- on nonterminal-free expressions, in particular every $\renderalt$ output, the four clauses alone determine it, and derivation trees are the standard derivations of $P$. This is the entire semantics used below; it depends only on the three disambiguation rules, not on engine internals.

\subsection{The Escaping Lemma (Lemma 1): standard substitution, one engine subtlety}
\label{app:escaping}
The Escaping Lemma is an instance of \textbf{standard grammar substitution}, and we
treat it as such rather than reprove it from scratch. Escaping renders each
$\esc{c_i}=\texttt{"}\,\mathrm{escape}(c_i)\,\texttt{"}$ a single indivisible terminal
(the escape is invertible, so $\mathcal{L}\proj{\esc{c_i}}=\{c_i\}$ --- lexical
atomicity of an escaped literal, standard \cite{aho2006compilers}), and context-free
languages are closed under substituting a terminal set into a production
\cite{hopcroftullman1979}; hence a parenthesized candidate alternation $e_\vee$ spliced
into a concatenative context $\alpha\,(\,e_\vee\,)\,\beta$ denotes exactly
$L(\alpha)\cdot C\cdot L(\beta)$. \texttt{render\_alternation} de-duplicates
order-preservingly ($C$ unchanged) and raises \texttt{GammaEmptyError} on the empty set
(Theorem~1's non-empty side-condition).

The one point that is \emph{not} routine --- because it is specific to GBNF's operator
precedence --- is \textbf{why the parentheses are mandatory}. Drop them, and
$\alpha\,\esc{c_1}\mid\cdots\mid\esc{c_k}\,\beta$ parses (\texttt{|} binding loosest) as
$(\alpha\,c_1)\mid\cdots\mid(c_k\,\beta)$, whose language is in general
\emph{incomparable} with $L(\alpha)\,C\,L(\beta)$ --- for $\alpha=$``\texttt{T.gemm(}'',
$\beta=$``\texttt{)}'', $k=2$, the two sets are disjoint --- silently splitting the
production and admitting stray fragments. The parentheses close this split; it is the
sole non-routine invariant of the splice, and it is exactly one of the cases the
differential test of \autoref{app:engine-assumption} pins.
\subsection{Assumption 1: differential-testing protocol}
\label{app:engine-assumption}
\paragraph{Assumption 1 (engine--denotation correspondence).} The minimal-GBNF denotation $\mathcal{L}\proj{\cdot}$ above agrees with the masking engine's (XGrammar's) actual GBNF parser on the fragments we instantiate. The syntactic risk surface is finite and explicitly enumerated: instantiated slot-carrying productions use exactly two GBNF constructs --- \emph{double-quoted string literals} and \emph{parenthesized alternation} --- with no regex and no repetition (the closed construct set fixed in \autoref{sec:formal-prelim}; slot-free productions may use richer constructs and contribute no references). We empirically validate this assumption by differential testing~\cite{mckeeman1998differential} over that complete construct vocabulary and adversarial representatives of its value space. Driving the executor's own instantiation path (\texttt{render\_alternation} $\to$ \texttt{from\_ebnf} $\to$ the engine's acceptance check), across \textbf{24 $\Gam\times$construct configurations} --- covering all five escape classes plus quote, backslash, newline, unicode, 256-character, and metacharacter ($\mid$, parentheses, quote) boundary names, single- vs multi-candidate, and the empty-$\Gam$ error path --- the engine \textbf{accepts all 90 in-set strings and rejects all 772 controlled out-of-set strings} (ghost names, near-misses, bare prefixes/suffixes, metacharacter injection, and \texttt{|}-split fragments). The engine's accept-set equals $\mathcal{L}\proj{\cdot}$ = the declared-name set exactly on these tests, with zero discrepancy (test \texttt{test\_assumption1\_differential}, 16 cases). Two cases pin the subtle points of this section directly: that the parentheses keep an alternation confined to the slot (no \texttt{|}-split escape), and that a name containing GBNF metacharacters is treated as a literal, not parsed as structure. The token-level clause of Assumption~1 (exact incremental recognizer) is additionally probed by bitmask-replay tests in the artifact (per-step admitted-token sets recomputed against the compiled automaton, ghost tokens asserted masked). These tests cover the complete construct vocabulary used by instantiated slots, but do not replace the engine-correspondence premise; Theorem~1 remains explicitly modulo Assumption~1.

\subsection{Proof of Proposition 3, and the bounded-name reading}
\label{app:prop3-proof}
\begin{definition}[Static masking system (formal; cf.\ \autoref{sec:construction} of the main text)]
Fix an unbounded identifier space $\mathrm{Id}$. A \textbf{static masking system} is a
finite set $\mathfrak{G}=\{G_1,\dots,G_m\}$ of grammars fixed before decoding --- each
contributing at a reference position a \emph{fixed} support set $S_i$, in practice a
finite name set or all of $\mathrm{Id}$; the proof uses only that the library is
finite and fixed (the \textbf{static-support assumption}: the support presented at a reference
position depends only on the dispatched grammar and production, so the
reference-support library is finite and fixed before decoding; closing it under
Boolean combinations still yields finitely many supports, so the pigeonhole below is
unaffected --- $\delta$ arbitrary, even uncomputable, is already the strongest
statically-supported adversary) --- together
with an \emph{arbitrary} dispatch function $\delta$ mapping prefixes to
$\{1,\dots,m\}$: $\delta$ may inspect the whole prefix, but may not synthesize a
support outside $\{S_1,\dots,S_m\}$. The system is \emph{sound} if at every reference
position its support contains no name outside the declared set, and
\emph{non-blocking} if it never masks out a scope-safe continuation.
\end{definition}
\begin{proof}[Full proof of Proposition 3 (main text)]
(ii) The declare-then-use hypothesis furnishes the witnesses: for each
$n\in\mathrm{Id}$, instantiating the form at $n$ and truncating at its use-site
yields a scope-safe prefix $p_D$ declaring exactly $D=\{n\}$, ending at a reference
position where $n$ extends to a complete scope-safe program (the form's own
completion) --- so the names that are scope-safe continuations there are exactly
$D$. The singleton family suffices for the pigeonhole; DSLs with chained
declarations realize every finite $D$ the same way. The system presents at $p_D$ one of finitely many supports drawn from the fixed
library; write $S(p_D)$ for it. Soundness at $p_D$ requires $S(p_D)\subseteq D$
(ruling out any all-of-$\mathrm{Id}$ support); non-blocking requires
$D\subseteq S(p_D)$ (every name in $D$ is a scope-safe continuation); together they
force $S(p_D)=D$. There are infinitely many distinct finite $D$ but only finitely
many eligible supports, so two prefixes $p_D\neq p_{D'}$ with $D\neq D'$ receive
the same support $S=D=D'$ --- contradiction. (Even the weakest liveness fails under soundness alone: the fresh-name singletons
$\{x_1\},\{x_2\},\dots$ are pairwise disjoint; soundness rules out any
all-of-$\mathrm{Id}$ support, so the pigeonhole over the finitely many remaining
supports yields two \emph{disjoint} $D,D'$ sharing one support $S$, and
$S\subseteq D\cap D'=\emptyset$ --- a total deadlock at that position.) (iii) ($\Leftarrow$) is Theorem~1 (modulo Assumption~1) plus Lemma~1's exact
instantiation with the pure declaration-set query ($\queryslot=\dom$), stated for
non-empty $\dom(\Gam(p))$ --- an empty environment admits no scope-safe
continuation, so the empty support is vacuously sound and non-blocking (the
implementation refuses such positions, \texttt{GammaEmptyError}).
($\Rightarrow$) soundness gives $S\subseteq\dom(\Gam(p))$; non-blocking gives the
reverse inclusion, since each declared name extends to a scope-safe completion
(scope-safety is name-level, so any skeleton completion over declared names
witnesses it).
\end{proof}
\paragraph{The bounded-name reading.}
If $\mathrm{Id}$ is artificially bounded to $n$ names, a static family of $2^{n}$
grammars (one per declaration set) evades (ii) --- at an exponential blow-up; this is
the pre-enumeration reading of \autoref{sec:construction} of the main text. The content of Proposition~3 is that the practical
impossibility is structural, not merely exponential: over the unbounded name spaces of
real DSLs, \emph{there is no finite static support family at all}. Prefix-dependent
reference support is therefore a necessary runtime object. Self-extension --- the
program extending the grammar that constrains its continuation --- is our
grammar-level realization of that object.

\subsection{The environment-scope guarantee boundary: full table}
\label{app:boundary}
\autoref{tab:decidability-boundary} gives the boundary in full: for each
$\Gam$-expressible property class (P1--P4) the enforcing mechanism and its
measured witness, and for each richer class (N1--N3) the additional analysis it
requires.
\begin{table}[!tp]
\centering
\small
\caption{The applicability boundary: $\Gam$-expressible properties (P1--P4) are guaranteed at construction; N1--N3 compose through a verifier, monitor, or verified support set, with coarse-rung fallback (\autoref{sec:sdcg}).}
\label{tab:decidability-boundary}
\begin{tabularx}{\textwidth}{@{}c >{\raggedright\arraybackslash}p{0.16\textwidth} c X X@{}}
\toprule
\# & Property & $\Gam$-slot? & Mechanism & Witness \\
\midrule
P1 & \textbf{scope safety} (reference $\in$ declared, no-ghost) & \checkmark & Theorem~1: $\Gam$-slot whitelist + Lemma~1 & gamma axis \textbf{0 NameError}; operands off-\Gam: 100\% ghosts; compile cliff 0$\to$100\% (main text: SQL \autoref{sec:eval-sql}, P4 \autoref{sec:eval-invariance}, GQA \autoref{sec:eval-roadmap}) \\
\addlinespace
P2 & \textbf{sort / type-tag match} & \checkmark$^{*}$ & registry asserts \texttt{fragment.sort == expected} + kind-filtered slot; cross-boundary use additionally requires a host$\leftrightarrow$embedded sort correspondence & GQA gamma operands kind-filtered, 0 sort violations \\
\addlinespace
P3 & \textbf{declaration-before-use} (\textbf{hole granularity}) & \checkmark & write-once live frames + slot evaluation against the active environment at hole entry. Restriction: declaration-introduction at the hole boundary, not statement granularity & Corollary~2 \\
\addlinespace
P4 & \textbf{shape agreement under a shape-typed slot} & \checkmark$^{\dagger}$ & shape into $\Gam$ attrs, operand slot shape-filtered: $\queryslot(\Gam)=\Gam.\mathrm{names}(\mathit{sort}=\mathtt{Buffer},\mathit{shape}=\sigma)$ $\to$ shape-mismatch candidates masked out & \autoref{sec:evaluation}, real 0.6B n=20: kind-only cross-shape mis-selection 2.5\% vs shape-typed 0\% by construction \\
\addlinespace
N1 & \textbf{semantic-value / numerical correctness} & $\times$ & semantic verifier or monitor (undecidable in general; Rice) & gamma numeric 62\% ($n{=}8$, main text \autoref{sec:eval-roadmap}): compilation and reference safety do not imply numerical correctness \\
\addlinespace
N2 & \textbf{general termination} & $\times$ & termination proof or verified support set (halting in general) & outside the monotone-$\Gam$ guarantee boundary \\
\addlinespace
N3 & \textbf{attribute-layer invariants} & $\times$ & stateful monitor (\texttt{update} may roll back \texttt{True->False}) & design-level declaration \\
\bottomrule
\end{tabularx}

\smallskip
\footnotesize
$^{*}$ \checkmark{} (single registry / hole, \emph{operational}); \textbf{conditional} across boundaries. \quad
$^{\dagger}$ \checkmark{} \textbf{by construction} (shape write-once); \textbf{binding-ness measured e2e}. \quad
N3: $\times$ not monotone.
\end{table}
\FloatBarrier

\subsection{Lifting and composing: attribute-typed environments, nesting, induction}
\label{app:composition}
\paragraph{Attr-typed \Gam{} (one dimension up).}
On the GQA kernel we tighten the shared-operand slot from a \texttt{kind} filter
(4 candidates) to a \texttt{kind+shape} filter (2 candidates), with zero engine
changes. \textbf{Cross-shape mis-selection (choosing K/V where a
$[\texttt{block\_M},\texttt{dim}]$ Q is required) becomes impossible by
construction} --- a deterministic adversarial test confirms K/V are masked out
of the support set. The guarantee is binding: on real 0.6B ($n=20$, 160
holes/arm) the kind-only arm mis-selects a wrong-shape buffer \textbf{2.5\%} of
the time, whereas the kind+shape arm is \textbf{0\% by construction}. The judge
is symbolic shape, independent of the runtime \texttt{block\_M=block\_N}
coincidence (a numeric oracle cannot catch a Q$\leftrightarrow$K swap when shapes
are equal). This lifts the guarantee from names to types; same-shape
mis-selection requires a separate mechanism.

\paragraph{Multi-grammar nested execution composes and stays correct.}
An outer Python region embeds an inner \texttt{@tilelang.jit} kernel region, which in a three-layer variant further embeds a loop-body region; the executor pushes and pops fragments to arbitrary depth. Qwen3-0.6B under this composition produces kernels that exec, CUDA-compile, and \textbf{numerically match} the torch reference (\texttt{max\_abs\_err} = 0.0157, 2/2 seeds). The same compile-plus-numeric result holds (2/2 seeds each) for the TagDispatch single-matcher region switch, a restricted-Python outer grammar, and a non-GEMM elementwise task. Separately (no torch oracle), the JSON-schema composition with a \Gam-dependent reference (declare-then-use) is a \textbf{construction-level reference-consistency = 100\%} result --- every emitted formula references only previously-declared variables, a property a plain JSON Schema cannot express~\cite{synchromesh2022iclr}.

\paragraph{Inducing the ladder across paradigms (SQL / jq / TileLang).}
The inductor auto-induces a 4-rung refinement ladder (\texttt{base}$\to$\texttt{gamma}$\to$\texttt{ctx}$\to$\texttt{pin}) on three DSLs of different paradigm, each LLM-proposed fragment validated by a \textbf{parser hard-gate that uses the same engine semantics as the executor mask} (the gate instantiates \texttt{\%SLOT\%} against \Gam{} and rejects ghost names --- the same mechanism that masks tokens at decode time). The three induced paradigms each carry a real oracle: \textbf{SQL columns} (induced over 25 real Spider samples, passing the hard-gate, equivalent to the hand-written fragment), \textbf{jq fields} (field $\in$ JSON keys; \texttt{gamma} answers 7/7 while \texttt{base} returns a ghost-field \emph{silent null} --- more insidious than SQL's crash), and \textbf{TileLang operands} (operand $\in$ declared buffers + shape). Stage-A2 anchor-based structural alignment (\autoref{sec:construction} of the main text) recovers \texttt{block\_M}/\texttt{block\_N}/\texttt{block\_K} across 3 GEMM variants at 100\% template coverage, and the stage-A4 per-tier \freebits{} score is monotone along the ladder (the quantitative face of Prop~1)~\cite{krogmeier2023oopsla}. A fourth domain, \textbf{library usage (AscendC)}, evaluates the static induction layer: from the API surface we learn the grammar mechanically (API names as an enumeration, template arguments as enumerations over the library's domains) and show a ghost API call is construction-level impossible. This case uses the parser hard-gate because the vendor compiler is unavailable locally (\autoref{sec:eval-induction} of the main text). Induction is grammar-as-distillation: the strong model distills programming-language-specific context understanding into the grammar / support set, not into weights.

The implementation carries a full test suite (\texttt{pytest -m 'not gpu'} + GPU end-to-end), including adversarial probes of the construction-level claims (32-deep nesting, 300-name \Gam{} sets, and maliciously injected names).

\subsection{Witnessing the refinement order on the ladder (full argument)}
\label{app:refinement}
\paragraph{Witnessing $\refn$ on the ladder; $\freebits$ monotonicity.} General $L$-inclusion is not cheaply decidable, and we never decide it. On our induction ladder (pin $\refn$ ctx $\refn$ gamma $\refn$ base) all rungs share one carrier production and differ only in a slot's candidate set, with $C_{\mathrm{tight}}\subseteq C_{\mathrm{loose}}$ --- so the projection functions of any two rungs are literally identical and projection alignment $\palign$ holds trivially: all three conjuncts of $\refn$ close on the ladder; by Lemma~1 the slot instantiates to language exactly $C$, so $\Lang{G_{\mathrm{tight}}}=L(\alpha)\,C_{\mathrm{tight}}\,L(\beta)\subseteq L(\alpha)\,C_{\mathrm{loose}}\,L(\beta)=\Lang{G_{\mathrm{loose}}}$ --- $L$-refinement is immediate, a cheap \emph{sufficient} condition, never a general decision. As a one-line consequence, a candidate \emph{subset} gives monotone non-increasing $\freebits$ as one tightens \emph{down} $\refn$ --- pointwise, per shared prefix: at each step the tighter rung's admitted-token set is a subset of the looser rung's, while the per-decode sums reported in \autoref{sec:evaluation} aggregate over each arm's own sampled paths --- so the experimental chain --- $\freebits$ decreasing from $998$ at the loosest rung, through $538$ and $287$, to $92$ at the tightest (\autoref{sec:evaluation}) --- is confirmation of the theory, not its source.

\subsection{The full non-context-freeness argument for Remark 1 (main text)}
\label{app:non-cf}
The claim --- ``reference $\subseteq$ declared, declaration-before-use'' is not
CFG-expressible --- follows by a closure argument in the tradition of Floyd's
ALGOL~60 proof \cite{floyd1962algol}. Names must range over an alphabet of at least
two letters --- a hypothesis this proof route needs: over a \emph{unary} name space
the minimal single-declare/single-use skeleton intersects to $\{a^n\#a^n\}$, which
\emph{is} context-free, so the copy-language reduction below has no unary
counterpart. (The hypothesis delimits the route, not the property: a unary skeleton
with one declaration and \emph{two} uses intersects $L_{\mathrm{scope}}$ to a
three-way-matched language of the $\{a^nba^nba^n\}$ kind, again non-context-free.)
Let $L_{\mathrm{scope}}$ be the declaration-consistent language
(\autoref{sec:formal-prelim} of the main text) over names in $\{a,b\}^{+}$, and suppose it were
context-free. Context-free languages are closed under intersection with regular
languages \cite{hopcroftullman1979}, and the skeleton
$R=\mathtt{decl}(w)\,;\,\mathtt{use}(w')$ with $w,w'\in\{a,b\}^{+}$ is regular;
with a single declaration in scope, declaration-consistency forces $w'=w$, so
$L_{\mathrm{scope}}\cap R=\{\mathtt{decl}(w)\,;\,\mathtt{use}(w):w\in\{a,b\}^{+}\}$;
the homomorphism erasing the fixed delimiters (wlog disjoint from $\{a,b\}$; otherwise apply a finite transduction, under which CFLs are likewise closed) and sending the separator to $\#$
maps this to the copy language $\{w\#w : w\in\{a,b\}^{+}\}$, and context-free
languages are closed under homomorphism \cite{hopcroftullman1979}, so the copy
language would be context-free --- yet it is not (textbook, via Ogden's lemma) ---
contradiction. Hence no
single CFG accepts exactly $L_{\mathrm{scope}}$; $\wallop$ sidesteps this by
re-instantiating per runtime $\Gam$, each instance a plain CFG.

\subsection{The environment-scope boundary (continued): P4's precise restriction}
\label{app:p4-boundary}
\paragraph{P4's precise restriction.} P4 is a \emph{construction-level} guarantee whose decidability rests on shape being a write-once constant (written once at \texttt{declare}, never rewritten by \texttt{update} --- caller discipline). Under that restriction shape is constant over the name's lifetime, the shape-filter is well-defined, and by the \emph{same} Lemma~1 + Theorem~1 argument the operand at that position has shape exactly $\sigma$. This is not attribute-layer monotonicity (N3 still holds) --- it is the controlled ``attribute as write-once constant'' sub-use; a rewritten attribute (e.g. \texttt{initialized}) falls under N3. Unlike P1's GQA $p<0.001$ backing, P4's binding-ness is measured end-to-end (\autoref{sec:evaluation}), showing the guarantee is non-vacuous (the model really does mis-select wrong-shape buffers when available) rather than a statistical claim (P4 is immune to sample size).

\subsection{Proof of the whole-program corollary}
\label{app:whole-program-proof}
\begin{proof}[Full proof of the whole-program corollary (main text)]
Fix a completed generation $w$ with grammar-call tree $T$ (finite and acyclic,
\autoref{sec:construction} of the main text). Any derivation of $w$ in the composed grammar factors
through the leaf fragments instantiated along $T$. For each reference occurrence
$r$, let $h(r)$ be its leaf hole and let $\Gam_r$ be the active environment when
that hole is entered. By Lemma~1, $r$'s yield lies in
$\queryslot(\Gam_r)$; by WF-query it lies in $\dom(\Gam_r)$. Thus
$\operatorname{name}(r)\in\dom(\Gam_r)$, proving hole-granularity
declaration-before-use. Later updates or frame pops affect later candidate sets,
not the earlier resolution. For references resolved to the persistent global layer,
name$\to$sort monotonicity additionally places them in
$\dom(\Gam^g_{\mathrm{final}})$. Composite fragments contribute reference
positions through their leaf fragments, which the hypothesis covers one by one.
\end{proof}

\fi
\end{document}